\documentclass[11pt]{article}
\usepackage[utf8]{inputenc}
\usepackage[margin=1.00in]{geometry}

\usepackage{dsfont}
\usepackage{enumitem}
\usepackage{graphicx}
\usepackage{xcolor}
\usepackage{sectsty}
\usepackage{enumitem}
\usepackage{caption}
\usepackage{algorithm}
\usepackage{algpseudocode}
\usepackage{appendix}
\usepackage{amsfonts}
\usepackage{amsmath}
\usepackage{amsthm}
\usepackage{amssymb}
\usepackage{epstopdf}
\usepackage[colorlinks=true,breaklinks=true,bookmarks=true,urlcolor=blue,citecolor=blue,linkcolor=blue]{hyperref}

\theoremstyle{definition}
\newtheorem{theorem}{Theorem}[section]
\newtheorem{proposition}{Proposition}[section]

\newtheorem{lemma}{Lemma}[section]

{\itshape}{\rmfamily}

\numberwithin{equation}{section}

\newcommand{\mca}{{\mathcal{A}}}
\newcommand{\mcb}{{\mathcal{B}}}
\newcommand{\mcg}{{\mathcal{G}}}

\def\bs{\boldsymbol}
\def\mbb{\mathbb}
\def\mcal{\mathcal}
\newcommand{\F}{\mathbb{F}} 
\newcommand{\Prob}{\mathbb{P}} 
\newcommand{\E}{\mathbb{E}} 
\newcommand{\val}{v} 

\def\kb{{\bar{k}}}
\newcommand{\argmax}[1]{\underset{#1}
{\operatorname{arg}\,\operatorname{max}}\;}
\newcommand{\argmin}[1]{\underset{#1}
	{\operatorname{arg}\,\operatorname{min}}\;}

\title{A General Lotto game with asymmetric budget uncertainty}

\author{Keith Paarporn, Rahul Chandan, Mahnoosh Alizadeh, Jason R. Marden}
\date{\small Department of Electrical and Computer Engineering, University of California Santa Barbara \\ {\tt\small kpaarpor@uccs.edu}, {\tt\small rchandan@ucsb.edu}, {\tt\small alizadeh@ucsb.edu}, {\tt\small jrmarden@ece.ucsb.edu} }

\begin{document}

\maketitle

\abstract{The General Lotto game is a popular variant of the famous Colonel Blotto game, in which two opposing players allocate limited resources over many battlefields. In this paper, we consider incomplete and asymmetric information formulations regarding the resource budgets of the players. In particular, one of the player's resource budget is common knowledge while the other player's is private. We provide complete equilibrium characterizations in the scenario where the private resource budget is drawn from an arbitrary Bernoulli distribution. We then show that these characterizations can be used to analyze a multi-stage resource assignment problem where a commander must decide how to assign resources to sub-colonels that compete against opponents in separate General Lotto games. While optimal deterministic assignments have been characterized in the literature, we broaden the context by deriving optimal (Bernoulli) randomized assignments, which induce asymmetric information General Lotto games to be played. We demonstrate that randomizing can offer a four-fold improvement in the commander's performance over deterministic assignments. }

\maketitle

\section{Introduction}\label{sec1}

Informational asymmetries play a significant role in determining the outcome of a competition. For example, political parties rely on leveraging voter data for election campaigning. Firms utilize customers' data and preferences to provide product recommendations and targeted advertisements. The firms may have incomplete and asymmetric information about the research or marketing capabilities of competing firms. Obtaining informational advantages is also a high priority in defense and intelligence operations. Clearly, how a competitor should strategize under asymmetric information environments is a central question that arises across a variety of domains. 

In this paper, we consider incomplete and asymmetric information formulations of a  competitive resource allocation model known as the General Lotto game \cite{Hart_2008,Kovenock_2020}. The General Lotto game is a popular variant of the classic Colonel Blotto game, in which two opposing players strategically allocate limited resources over a number of valuable battlefields \cite{Borel,Gross_1950,Roberson_2006,Golman_2009,Kovenock_handbook_2012,Kovenock_2020}. The objective for both players is to accumulate as much value as possible by securing battlefields. In order to secure an individual battlefield, one must send a higher amount of resources to that battlefield than the other player. In our formulations, we consider a General Lotto game where one of the player's resource budget is drawn from a Bernoulli distribution, and the other player is uninformed about the true realization. The primary contribution of this paper is the full derivation of equilibrium payoffs and strategies to this incomplete and asymmetric information game. 

We then apply these equilibrium characterizations to a multi-stage resource assignment problem originally proposed by Kovenock and Roberson in \cite{Kovenock_2012}, with several subsequent studies based on this model \cite{Gupta_2014a,Gupta_2014b,Chandan_2020,Chandan_Concession}. Here, a high-level commander faces two opponents on two separate fronts of battlefields. The commander's task is to assign limited or costly resources to two independent sub-colonels, wherein they locally engage in a General Lotto game against its opponent. The objective of the commander is to maximize the cumulative payoffs that the sub-colonels achieve in equilibrium from their respective General Lotto games. The problem of how to assign resources among multiple fronts is highly relevant in a multitude of contexts, given that a single centralized decision-maker is often unable to make entire allocation decisions due to physical or computational constraints. Indeed, many large-scale organizations operate with  hierarchical decision-making structures, where an executive or director delegates authority or endows resources to local decision-makers. 

It is important to note that the commander's assignment strategy determines the informational environment of the local General Lotto games. In particular, a deterministic assignment will induce a complete information game, wherein the opponents know the sub-colonels' assigned resources.  This is the setting featured in Kovenock's original model \cite{Kovenock_2012}. On the other hand, a randomized assignment induces one-sided incomplete and asymmetric information  modeled by our formulations. By leveraging our equilibrium characterizations, we derive the optimal \emph{randomized} assignments, and demonstrate that they can offer a four-fold performance improvement over optimal deterministic assignments.


\subsection{Contributions}

A graphical summary of our main contributions is given in Figure \ref{fig:setup}. A more detailed summary is given below. 

\vspace{2mm}

\noindent{\bf Equilibrium characterizations of Bernoulli Lotto games.} To the best of our knowledge, General Lotto games with one-sided budget uncertainty have not been considered in the existing literature. Our primary contributions in this paper are complete equilibrium characterizations for \emph{Bernoulli General Lotto} games (Theorem \ref{thm:BNE_regions}) -- one player has private information about its budget endowment drawn from a Bernoulli distribution, while the other player's budget endowment is common knowledge. We emphasize the novelty of these solutions: primarily, games with \emph{symmetric} private information on resource budgets have been considered in literature \cite{Adamo_2009}.

\begin{figure}
	\centering
	\includegraphics[scale = .15]{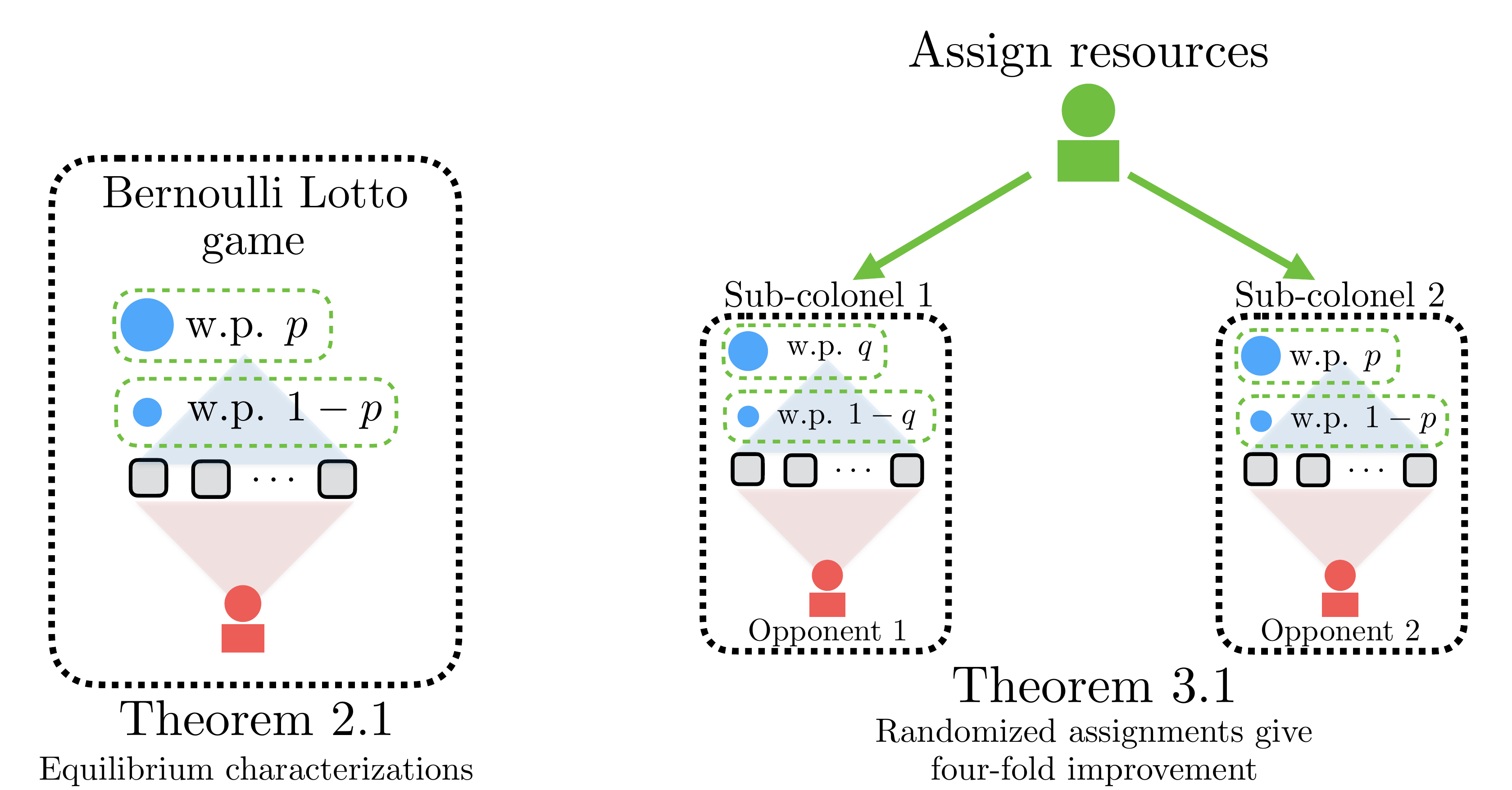}
	\caption{ (Left) In the Bernoulli Lotto game, one player's resource budget is drawn at random from a Bernoulli distribution and becomes private information. The other player's budget is common knowledge. A main contribution of this paper is the complete equilibrium characterizations under this asymmetric information environment (Theorem 2.1). (Right) The resource assignment problem. Here, a commander is responsible for assigning resource budgets to two sub-colonels, who subsequently engage in a Lotto game with its respective opponent. Randomized assignments induce asymmetric information environments between sub-colonels and opponents. The second main contribution of this paper is the complete characterization of optimal randomized assignments, and quantifying a four-fold improvement over deterministic assignments.}
	\label{fig:setup}
\end{figure}

We then compare the performance of the player with the randomized budget to a benchmark complete information scenario in which it uses its average budget deterministically (Section \ref{sec:sims}). Due to the randomized budget constraints placed in the Bernoulli Lotto game and the linearity of expectations, it is immediate that the player with randomized budget performs \emph{worse} compared to the corresponding benchmark scenario (Proposition \ref{prop:CI_better}). As such, we investigate the extent to which a randomized budget hurts performance in comparison to the benchmark scenario, through analytical and numerical methods (Proposition \ref{prop:zeta_p})\footnote{An alternate comparison in which the uninformed player receives a noisy signal about the private budget type may also be considered. Compared to this scenario, it must be the case that the player with private budget type performs better when no signal is sent at all. We leave a full investigation of such comparisons to future work.}. This observation does not hold when there are marginal costs associated with expending resources (instead of use-it or lose-it budgets), which we establish in our results on the resource assignment problem.

\vspace{2mm}

\noindent{\bf The application to multi-stage resource assignment problems.} We apply our equilibrium solutions to Bernoulli Lotto games to the resource assignment problem first proposed by Kovenock \cite{Kovenock_2012}. The assignment problem unfolds in the following stages.

\noindent{$\bullet$} \textbf{Stage 1}: A commander assigns limited or costly resources to two sub-colonels $\mca_1$ and $\mca_2$ that are to engage in competition with opponents $\mcb_1$ and $\mcb_2$, respectively.

\noindent{$\bullet$} \textbf{Stage 2}: The opponents $\mcb_1$ and $\mcb_2$ observe the commander's assignment strategy, and each decide how many costly resources to invest in.

\noindent{$\bullet$} \textbf{Stage 3}: Sub-colonel $\mca_i$ ($i = 1,2$) engages in a General Lotto game with opponent $\mcb_i$, using resource budgets determined from the decisions in stage 1 and 2. 

The objective of the commander is to maximize the resulting cumulative payoff awarded to the sub-colonels in their respective General Lotto games. The commander's assignment strategy in stage 1 can be classified as either \emph{deterministic} or \emph{randomized}. In a deterministic assignment, the sub-colonels' assigned budgets become common knowledge, and then two (complete information) General Lotto games are played according to their equilibrium outcomes in stage 3. Our equilibrium characterizations to Bernoulli Lotto games enable us to analyze classes of \emph{randomized} assignment strategies. That is, if the commander's assignment strategy is a Bernoulli distribution on resource assignments, the opponents remain uncertain about the sub-colonels' true endowments. This induces asymmetric information environments in the subsequent stage 3 General Lotto games.


This necessitates the question: how much can the commander improve its performance by randomizing its assignment compared to the optimal deterministic assignments? Our contributions here establish that the commander can obtain up to a four-fold performance improvement by using randomized assignment strategies in comparison to deterministic assignment strategies, in scenarios where utilizing resources is costly on both sides (Theorem \ref{thm:W_theorem}).


\subsection{Related works}

The primary literature on Colonel Blotto and General Lotto games focuses on deriving equilibria under the assumption of complete information. That is, all parameters -- player budgets and battlefield valuations -- are common knowledge. Deriving equilibria to Colonel Blotto games is known to be a fundamentally challenging task. Several works over the last 100 years have advanced the level of generality for equilibrium solutions \cite{Borel,Gross_1950,Roberson_2006,Kvasov_2007,Roberson_2012,Schwartz_2014,Macdonell_2015,Thomas_2018,Kovenock_2020}. However, analytical solutions to arbitrary settings of Colonel Blotto games are still an open problem. For this reason, researchers often study the General Lotto game, a relaxation of the Colonel Blotto game where players' budget expenditure only needs to hold in expectation \cite{Bell_1980,Myerson_1993,Hart_2008,Kovenock_2020,Vu_EC2021}. It is more analytically tractable while maintaining essential aspects of competitive resource allocation. The General Lotto game, along with several other variants, is often adopted to study more complex scenarios such as network security \cite{Fuchs_2012,Shahrivar_2014,Guan_2019,Shishika_2021} and the security of cyber-physical systems \cite{Gupta_2014a,Ferdowsi_2020}.

Few contributions in the Blotto and Lotto literature have shifted the focus away from complete information settings \cite{Adamo_2009,Ewerhart_2021,Paarporn_2022_LCSS}. This is in contrast to a significant body of work on incomplete information all-pay auctions \cite{Siegel_2014,Szech_2011,Einy_2017,amann1996asymmetric}. There is an intimate connection between the structures of equilibrium strategies in all-pay auctions and Colonel Blotto and General Lotto games \cite{Roberson_2006,Kovenock_2020,Vu_EC2021}. The primary difference is that in Blotto and Lotto games, players have a fixed, use-it-or-lose-it resource budget at their disposal. In an all-pay auction, competitors pay a per-unit cost for their investment decisions with no upper bound on how much they can bid.

The first main result of this manuscript contributes to a growing literature on Blotto and Lotto games with incomplete information. Similar to our budget uncertainty setting, the works \cite{Adamo_2009} and \cite{Kim_2017} study settings where players have incomplete information about each other's resource endowments. That is, the players' resource endowments are random variables and each player holds private information only about their own endowment. The model of \cite{Adamo_2009} consider the players to be equally uninformed (no information asymmetry) about the budget parameters, and hence symmetric Bayes-Nash equilibria are identified. The work of  \cite{Kim_2017} provides pure-strategy equilibria to an asymmetric information Lottery Blotto game, a variant that generally admits pure  equilibria \cite{Friedman_1958,Robson_2005,Kovenock_2019_Lottery,Vu_2020thesis}. Uncertainty about battlefield valuations has also been featured recently in the literature. In \cite{Kovenock_2011,Ewerhart_2021}, each player's valuations of the battlefields are independently drawn from a common joint distribution. Here, players are equally uninformed about the other players' valuations, and hence symmetric Bayes-Nash equilibria are identified. Asymmetric information about the battlefields' valuations is investigated in \cite{Paarporn_2019,Paarporn_2022_LCSS}.


The second main result of this manuscript pertains to a well-studied multi-stage resource assignment problem in the Blotto and Lotto literature \cite{Kovenock_2012,Gupta_2014a,Gupta_2014b,Chandan_2020}. Kovenock and Roberson \cite{Kovenock_2012} first formulated the framework of the assignment problem in the context of Colonel Blotto games \cite{Kovenock_2012}, where they derived optimal \emph{deterministic} assignment strategies against two opponents. Our work broadens this framework to randomized assignment strategies. Indeed, the focus of \cite{Kovenock_2012} was to determine when opponents could benefit from unilateral transfers of resources.  Several follow-up studies based on this framework have since appeared. The work in \cite{Gupta_2014a,Gupta_2014b} considered a setting where the two opponents can decide to add battlefields in addition to transferring resources amongst each other. Counter-intuitively, their model showed that the team of opponents can achieve better overall performance if the transfers are made public. Alternatively, the work of \cite{Chandan_2020,Chandan_Concession} studied settings where one of the opponents pre-commits resources before play, which effectively can alter the optimal assignment strategy such that the opponent benefits in subsequent interactions. In these works, deterministic assignments are considered and hence randomized assignments do not play a role. Going beyond the one-vs-two framework that is common in the Blotto literature, \cite{cortes2018generalising} considers generalized multi-player conflicts on arbitrary networks, where winning probabilities are determined by a contest success function \cite{skaperdas1996contest,Robson_2005}.

\section{Equilibrium characterizations}

We first introduce the classic General Lotto game with complete information. We then formulate our model of one-sided incomplete and asymmetric information on resource budgets, which we term ``Bernoulli General Lotto" games.

\subsection{General Lotto games}

A (complete information) General Lotto game consists of two players, $\mca$ and $\mcb$. Each player is tasked with allocating their resource budgets $A, B > 0$ across a set of $n$ battlefields. Each battlefield has an associated value $\val_j > 0$, $j\in[n]$.  An allocation for $\mca$ is any vector $\bs{x}_{\mca} \in \mbb{R}_{\geq 0}^n$, and similarly for $\mcb$. An admissible strategy for $\mca$ is a randomization $F_\mca$ over allocations such that the expended resources do not exceed the budget $A$ in expectation. Specifically, $F_\mca$ is an $n$-variate (cumulative) distribution that belongs to the family
\begin{equation}\label{eq:LC}
    \F(A) \triangleq \left\{ F : \E_{\bs{x}_{\mca}\sim F}\left[\sum_{j=1}^n x_{\mca,j}\right] \leq A \right\}.
\end{equation}
and similarly, $F_\mcb \in \F(B)$. Given a strategy profile $(F_\mca,F_\mcb)$, the utility of player $\mca$ is 
\begin{equation}\label{eq:Lotto_payoff}
    u_\mca(F_\mca,F_\mcb) \triangleq \E_{\substack{\bs{x}_{\mca}\sim F_\mca \\ \bs{x}_{\mcb}\sim F_\mcb}}\left[\sum_{j=1}^n v_j \cdot  \mathds{1}\{x_{\mca,j} > x_{\mcb,j} \}  \right]
\end{equation}
where $\mathds{1}\{\cdot\}$ is 1 if the statement in the bracket is true, and 0 otherwise\footnote{An arbitrary tie-breaking rule may be selected, without changing our results. This is generally true in General Lotto games \cite{Kovenock_2020}. For simplicity, we will assume ties are awarded to player $\mcb$.}. It follows that the utility of player $\mcb$ is
\begin{equation}
    u_\mcb(F_\mca,F_\mcb) \triangleq \|\bs{v}\|_1 - u_\mca(F_\mca,F_\mcb)
\end{equation}
An instance of the complete information General Lotto game is denoted by $\text{GL}(A,B;\bs{v})$. The equilibrium characterization of General Lotto games is well-established in the literature \cite{Hart_2008,Kovenock_2020}. The equilibrium payoff to $\mca$ in $\text{GL}(A,B;\bs{v})$ is given as follows:
\begin{equation}\label{eq:Lotto_CI}
	\pi^{\text{CI}}_\mca({A},B;\bs{v}) \triangleq \|\bs{v}\|_1\cdot
    	\begin{cases}
    	\frac{{A}}{2B}, &\text{if }  {A} < B \\
    	1 - \frac{B}{2{A}}, &\text{if } {A} \geq B \\
	\end{cases}
\end{equation}
and the equilibrium payoff of player $\mcb$ is simply $\|\bs{v}\| - \pi_\mcb^\text{CI}({A},B;\bs{v})$.

\subsection{Bernoulli General Lotto games}

Before play, the budget of player $\mca$ is drawn according to a Bernoulli distribution. With probability $p \in [0,1]$, player $\mca$ is endowed with a high budget $A^h \geq 0$, and with probability $1-p$, it is endowed with a low budget $A^\ell \geq 0$, where $A^\ell \leq A^h$. We denote this Bernoulli distribution as $\Prob_\mca = (A^h,A^\ell,p)$. The realized budget is private information to player $\mca$. An admissible action is thus a pair of strategies $\vec{F}_\mca = \{F^h,F^\ell\} \in \F(A^h) \times \F(A^\ell)$, where $F^h$ or $F^\ell$ is implemented conditional on which budget is realized. Player $\mcb$ does not observe the true realization, and thus selects a single strategy $F_\mcb \in \F(B)$ to implement regardless of which of $\mca$'s budget is realized. Without loss of generality moving forward, we assume the total sum of battlefield values is normalized to $\|\bs{v}\|_1 = 1$. Given a strategy profile $(\vec{F}_\mca,F_\mcb)$, the ex-ante expected utility to $\mca$ and $\mcb$ are:
\begin{equation}\label{eq:Bayesian_payoff}
    \begin{aligned}
        U_\mca(\vec{F}_\mca,F_\mcb) &\triangleq p\cdot u_\mca(F_\mca^h,F_B) + (1-p)\cdot u_\mca(F_\mca^\ell,F_B) \\
        U_\mcb(\vec{F}_\mca,F_\mcb) &\triangleq 1 - U_\mca(\vec{F}_\mca,F_\mcb)
    \end{aligned}
\end{equation}
We will refer to the simultaneous-move game with the above expected utilities as a \emph{Bernoulli General Lotto game}, and denote an instance with $\text{BL}(\Prob_\mca,B)$. In the case that the support of $\Prob_\mca$ is a singleton, it becomes a game of complete information.  When this holds, we will simply refer to this setting as a \emph{General Lotto} game, denoted $\text{GL}(A,B)$. We drop the vector of battlefield values $\bs{v}$ from the set of parameters for compactness. A strategy profile $(\vec{F}_\mca^*,F_\mcb^*)$ is an \emph{equilibrium} of $\text{BL}(\Prob,B)$ if
\begin{equation}
	\begin{aligned}
		U_\mca(\vec{F}_\mca^*,F_\mcb^*) \geq U_\mca(\vec{F}_\mca,F_\mcb^*) \quad \text{and} \quad
		U_\mcb(\vec{F}_\mca^*,F_\mcb^*) \geq U_\mcb(\vec{F}_\mca,F_\mcb)
	\end{aligned}
\end{equation}
for any $\vec{F}_\mca \in \F(A^h) \times \F(A^\ell)$ and $F_\mcb \in \F(B)$.

\subsection{Equilibrium payoffs in Bernoulli General Lotto games}\label{sec:BL}


We present the unique equilibrium payoffs for every instance of Bernoulli General Lotto games. Throughout, we will denote $\bar{A} = pA^h + (1-p)A^\ell$ as the expected resource endowment of player $\mca$ associated with $\Prob_\mca$. Player $\mcb$ has a resource budget of $B \geq 0$.   The complete characterization of the equilibrium payoffs is given in the result below. 


\begin{theorem}\label{thm:BNE_regions}
	Consider any instance of the Bernoulli General Lotto game, $\text{BL}(\Prob_\mca,B)$. The equilibrium payoff to player $\mca$ is
	\begin{equation}\label{eq:piA}
	    \begin{aligned}
    	&\pi_\mca(\Prob_\mca,B) \triangleq \\ 
    	&
    	\begin{cases}
    		\frac{\bar{A}}{2B}, &(A^h,A^\ell) \in \mathcal{R}_1 \\
    		1-\frac{B}{2\bar{A}}, &(A^h,A^\ell) \in \mathcal{R}_2 \\
    		p + (1-p)\left(1 - \frac{B}{2 A^\ell}\right), &(A^h,A^\ell) \in \mathcal{R}_3 \\
    		p + (1-p)\frac{A^\ell}{2B}, &(A^h,A^\ell) \in \mathcal{R}_4 \\
    		p + (1-p)\frac{A^\ell}{A^h} + \frac{\sqrt{\bar{A}(\bar{A} - pA^h)}}{A^h} - \frac{B\left(\sqrt{(1-p)A^\ell} + \sqrt{\bar{A}}\right)^2}{2 (A^h)^2}, &(A^h,A^\ell) \in \mathcal{R}_5 \\
    	\end{cases}
    	\end{aligned}
	\end{equation}
	where $\mathcal{R}_k$, $k=1,\ldots,5$, are disjoint subsets of $\mcal{R} \triangleq \{(A^h,A^\ell)\in\mathbb{R}_{\geq 0}^2 : A^h \geq A^\ell\}$, given by
	\begin{equation}\label{eq:regions}
		\begin{aligned}
			\mathcal{R}_1 &\triangleq \left\{ (A^h,A^\ell)\in\mcal{R} : \bar{A} \leq B \right\} \setminus \mcal{R}_5 \\
			\mathcal{R}_2 &\triangleq \left\{ (A^h,A^\ell)\in\mcal{R} : \bar{A} \geq B \text{ and } A^\ell \geq \frac{1-p}{2-p}A^h \right\} \\
			\mathcal{R}_3 &\triangleq \left\{ (A^h,A^\ell)\in\mcal{R} : A^h \geq \left(2 + \frac{p}{1-p}\right) B \text{ and } 1 \leq A^\ell \leq \frac{1-p}{2-p}A^h \right\} \\
			\mathcal{R}_4 &\triangleq \left\{ (A^h,A^\ell)\in\mcal{R} : A^h \geq \left(2 + \frac{p}{1-p}\right) B \text{ and } \frac{pB^2}{(1-p)(A^h - 2B)} \leq A^\ell \leq B \right\} \\
			\mathcal{R}_5 &\triangleq \left\{ (A^h,A^\ell)\in\mcal{R} : A^\ell \leq B\cdot H(A^h/B)  \right\} \\
		\end{aligned}
	\end{equation}
	with the function $H(a)$ defined as
	\begin{equation}\label{eq:G}
		H(a) \triangleq
		\begin{cases}
			0, &\text{if } a \in [0,1) \\
			\frac{p(a-1)^2}{(1-p)(2-a)}, &\text{if } a \in [1,2-p] \\ \frac{1-p}{2-p}a, &\text{if } a \in \left(2-p,2+\frac{p}{1-p}\right] \\ \frac{p}{(1-p)\left( a-2\right)}, &\text{if } a > 2 + \frac{p}{1-p}
		\end{cases}
	\end{equation}
	The equilibrium payoff to player $\mcb$ is given by $\pi_\mcb(\Prob_\mca,B) = 1 - \pi_\mca(\Prob_\mca,B)$.
\end{theorem}
Observe that the equilibrium payoffs do not depend on the values of any individual battlefields, and the result simply generalizes to $\|\bs{v}\|_1 \cdot \pi_\mca$ when the total value of all battlefields is not normalized to 1. 


We note that the payoff $\pi_\mca$ in regions $\mcal{R}_1$ and $\mcal{R}_2$ coincides with the payoff of the corresponding benchmark complete information game $\text{GL}(\bar{A},B)$ \eqref{eq:Lotto_CI} in which player $\mca$ utilizes its expected endowment $\bar{A}$ deterministically. Indeed, these payoffs are well-known from the literature \cite{Hart_2008,Kovenock_2020}.

Figure \ref{fig:regions} depicts the five disjoint regions of the parameter space for which \eqref{eq:piA} is defined. We devote Appendix \ref{sec:APA} to the proof of Theorem \ref{thm:BNE_regions}, which also details the players' equilibrium strategies. To establish the result, we begin by first deriving a connection to two-player all-pay auctions with incomplete and asymmetric information, and leverage\footnote{The equilibrium solutions of complete information all-pay auctions have been leveraged to completely characterize the equilibria of complete information General Lotto games \cite{Roberson_2006,Kovenock_2020,Vu_EC2021}} their equilibrium strategies  \cite{Siegel_2014} to derive a system of non-linear equations associated with the expected budget constraints \eqref{eq:LC}. We detail all solutions to this system in Proposition \ref{prop:SB_equil}, where we find that they can exist \emph{only for a subset of all possible underlying parameters}; in particular the sub-region $\mcal{R}_5$. We thus utilize other analytical techniques to fully characterize equilibria for all remaining parameters. We identify four disjoint regions $\mcal{R}_1$ - $\mcal{R}_4$ that constitute the remaining set of parameters, where each admits distinct structures in the equilibrium strategies (Appendix \ref{sec:other_regions}). The uniqueness of equilibrium payoffs follows from the BL game being constant-sum in ex-ante utilities.


%

\begin{figure}[t]
	\centering
	\includegraphics[scale = .2]{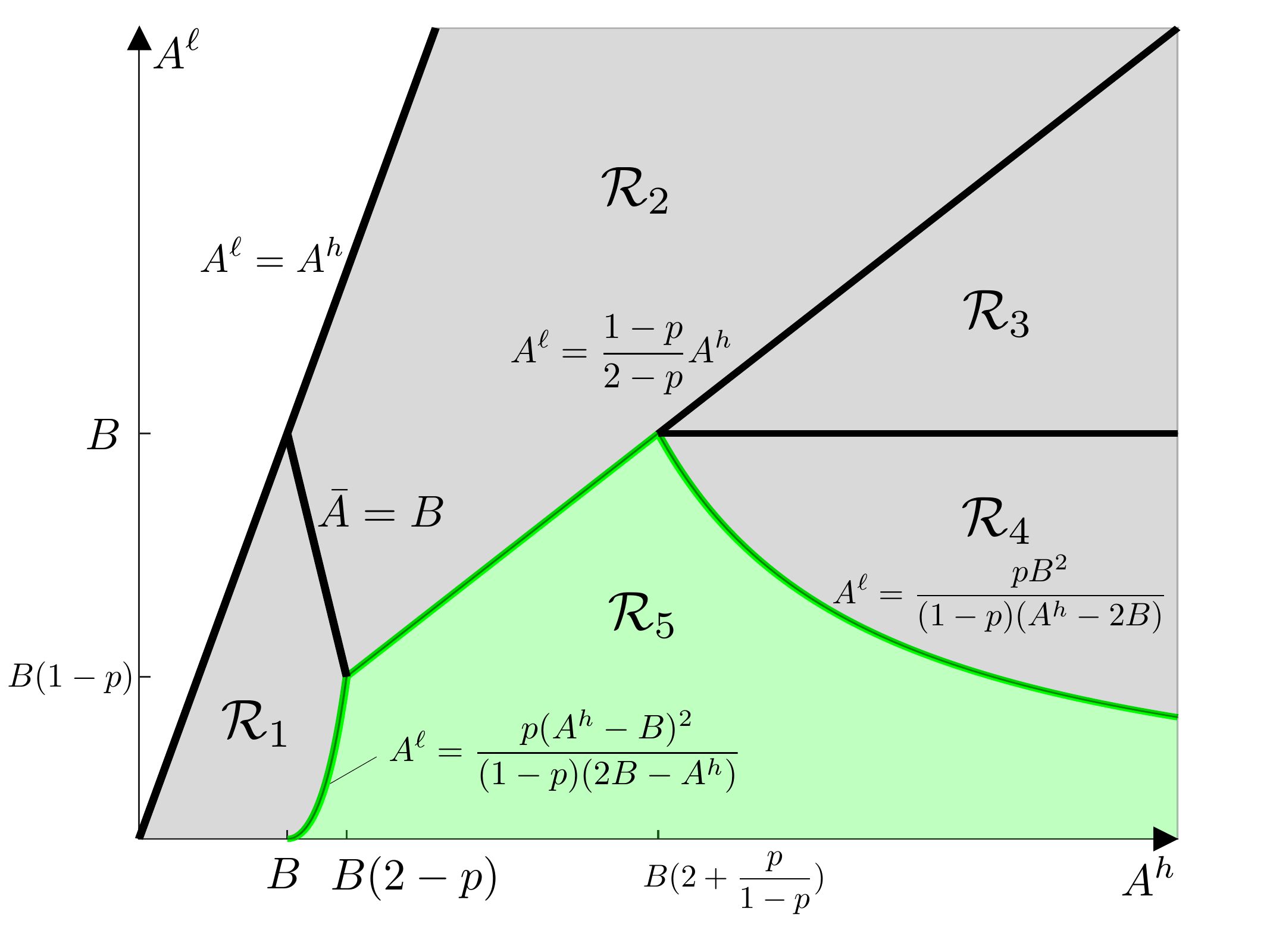}
	\caption{A diagram showing the five distinct parameter regions \eqref{eq:regions} that encompass the entire class of Bayesian Lotto games $\text{BL}(\Prob_\mca,B)$. Shown here is the space $A^h \geq A^\ell$ for a fixed $p$ and $B$.}
	\label{fig:regions}
\end{figure}

\subsection{Bernoulli Lotto games with per-unit costs}\label{sec:BL_costs}

In the preceding section, the equilibrium payoffs for Bernoulli Lotto games with fixed use-it or lose-it budgets were completely characterized. It is common in the literature to consider, instead of fixed budgets, a per-unit cost for players to expend resources. In the context of Bernoulli Lotto games, we consider each player having no budget limit on resources it can expend, but now impose cost terms in the utility functions. Specifically, with probability $p$, player $\mcal{A}$ has a cheap cost, $c^\text{ch} > 0$, and with probability $1-p$ has an expensive cost, $c^\text{ex}>0$, for which $c^\text{ch} \leq c^\text{ex}$. Player $\mcal{B}$ has a single cost type where its cost is $c > 0$. A feasible strategy for player $\mcal{A}$ is a pair $\vec{F}_\mcal{A} = \{F^\text{ch},F^\text{ex}\}$, where $F^\text{ch},F^\text{ex}$ are any distributions on $\mbb{R}_{\geq 0}^n$. Then, the ex-ante expected utility for player $\mca$ is now of the form
\begin{equation}
    \hat{U}_\mca(\vec{F}_\mca,F_\mcb) \triangleq U_\mca(\vec{F}_\mca,F_\mcb) - pc^\text{ch}\cdot\E_{F_A^\text{ch}}[\|\bs{x}_\mca\|_1] - (1-p)c^\text{ex}\cdot \E_{F_A^\text{ex}}[\|\bs{x}_\mca\|_1]
\end{equation}
Similarly, the ex-ante expected utility for player $\mcb$ is given by
\begin{equation}
    \hat{U}_\mcb(F_\mcb,\vec{F}_\mca) \triangleq U_\mcb(F_B,\vec{F}_\mca) - c\cdot\E_{F_\mcb}[\|\bs{x}_\mcb\|_1]
\end{equation}
where $U_\mca$ and $U_\mcb$ were defined in \eqref{eq:Bayesian_payoff}. The above utilities define a simultaneous-move game. We refer to an instance of this game as $\text{BLC}((c^\text{ch},c^\text{ex},p),c)$. 

Interestingly, it turns out that $\text{BLC}((c^\text{ch},c^\text{ex},p),c)$ coincides with $n$ simultaneous and independent two-player all-pay auctions with incomplete and asymmetric information, which have extensively been studied in the literature\footnote{The solution can be broken up into $n$ independent all-pay auctions, each corresponding to an individual battlefield. The associated cost for each player is divided by the valuation $v_j$.}. Indeed, \cite{Szech_2011} provides a more general solution in which both players can have up to two distinct cost types. Recognizing the structure between all-pay auctions and the BLC game, the characterization of of equilibrium payoffs is given below.

\begin{theorem}[\cite{Szech_2011,Siegel_2014}]
    Consider any instance of the Bernoulli Lotto game with costs, $\text{BLC}((c^\text{ch},c^\text{ex},p),c)$. The equilibrium payoff to player $\mca$ is
    \begin{equation}
        \begin{cases}
            0, &\text{if } p\frac{c^\text{ch}}{c} \geq 1 \\
            p(1-c^\text{ch}(\frac{p}{c} + \frac{1}{c^\text{ex}}(1 - \frac{c^\text{ch}}{c}p)), &\text{if } p\frac{c^\text{ch}}{c} < 1 \leq p\frac{c^\text{ch}}{c} + (1-p)\frac{c^\text{ex}}{c} \\
            1 - p\frac{c^\text{ch}}{c}, &\text{if } p\frac{c^\text{ch}}{c} + (1-p)\frac{c^\text{ex}}{c} < 1
        \end{cases}
    \end{equation}
    The equilibrium payoff to player $\mcb$ is
    \begin{equation}
        \begin{cases}
            1-\frac{c}{c^\text{ch}}, &\text{if } p\frac{c^\text{ch}}{c} \geq 1 \\
            1-p - \frac{c}{c^\text{ex}}(1-p\frac{c^\text{ch}}{c}), &\text{if } p\frac{c^\text{ch}}{c} < 1 \leq p\frac{c^\text{ch}}{c} + (1-p)\frac{c^\text{ex}}{c} \\
            0, &\text{if } p\frac{c^\text{ch}}{c} + (1-p)\frac{c^\text{ex}}{c} < 1
        \end{cases}
    \end{equation}
\end{theorem}

Even though the equilibrium solutions of $\text{BLC}((c^\text{ch},c^\text{ex},p),c)$ are immediate from the literature, we remark that there is \emph{not} a one-to-one correspondence of cost parameters $(c^\text{ch},c^\text{ex},c)$ to all possible tuples of expected budgets $(A^\text{ch},A^\text{ex},B) \in \mbb{R}_{\geq 0}^3$ expended in equilibrium. Recalling the diagram of Figure \ref{fig:regions}, the budget tuples that correspond to equilibrium solutions of all-pay auctions only arise in the green region, $\mcal{R}_5$. These results can also be derived directly using the methods outlined in Appendix \ref{sec:APA}. In the $\mcal{R}_5$ region, the high and low budget types are further apart compared to the $\mcal{R}_1$ and $\mcal{R}_2$ regions.

\section{Multi-stage resource assignment problem}\label{sec:commander_assignment}

In this section, we apply our equilibrium solutions of BL games to address a hierarchical resource assignment problem, which we term the ``commander assignment problem". 

\subsection{Commander assignment problem}

Here, we present a generalized formulation of the assignment problem originally proposed by Kovenock and Roberson \cite{Kovenock_2012}. A commander with a resource budget $\bar{A}_C > 0$ is responsible for assigning resources to two sub-colonels $\mca_1,\mca_2$ engaged in separate competitions against respective opponents $\mcb_1,\mcb_2$. The opponents have limited resource budgets $\bar{B}_1,\bar{B}_2 > 0$. The interaction unfolds in the following three-stage setup, which is also illustrated in Figure \ref{fig:CAP_stages}.

\noindent\textbf{Stage 1:} The commander chooses an assignment strategy $\Prob$, which is a distribution on the sub-colonels' resource assignments $(A_1,A_2) \in \mathbb{R}_{\geq 0}^2$. The strategy becomes common knowledge to the sub-colonels and the opponents. The commander pays the cost $c\cdot \E_\Prob[A_1 + A_2]$, where $c \geq 0$ is its per-unit cost for expenditure.

We will consider probability distributions $\Prob$ on assignments $(A_1,A_2) \in \mathbb{R}_{\geq 0}^2$ such that each marginal $\Prob_i$, $i=1,2$ has at most two values in its support\footnote{This class of randomized assignment strategies is considered due to all of our results for the underlying incomplete information General Lotto games applying to two-type budget uncertainty. It is of interest to extend these results to more than two randomized budget levels. We note there are analytical challenges associated with these extensions, which we detail in Appendix \ref{sec:APA}.}. Furthermore, we will consider assignment strategies that satisfy the commander's budget \emph{in expectation}:
\begin{equation}
	\mathbb{E}_\Prob[A_1+A_2] = \sum_{(A_1,A_2) \in \text{supp}(\Prob)} \Prob(A_1,A_2) \cdot (A_1 + A_2) \leq \bar{A}_C.
\end{equation}
Let us denote $\mcal{P}(\bar{A}_C)$ as the set of all such distributions $\Prob$ that satisfy the above condition.  As such, the marginal distribution $\Prob_i$ on assignments to sub-colonel $i$ is associated with a probability $p_i := \Prob_i(A_i^h)$ on a high endowment $A_i^h$, and the probability $1-p_i$ on a low endowment $A_i^\ell$, where $A_i^h \geq A_i^\ell$. We will often write the marginal as a tuple, $\Prob_i = (A_i^h,A_i^\ell,p_i)$, to explicitly represent this Bernoulli distribution. We say that $\Prob$ is a \emph{randomized} assignment if the support of at least one of its marginals is not a singleton, and that it is \emph{deterministic} if the supports of both marginals are singletons.

\vspace{1mm}

\begin{figure}
    \centering
    \includegraphics[scale=0.15]{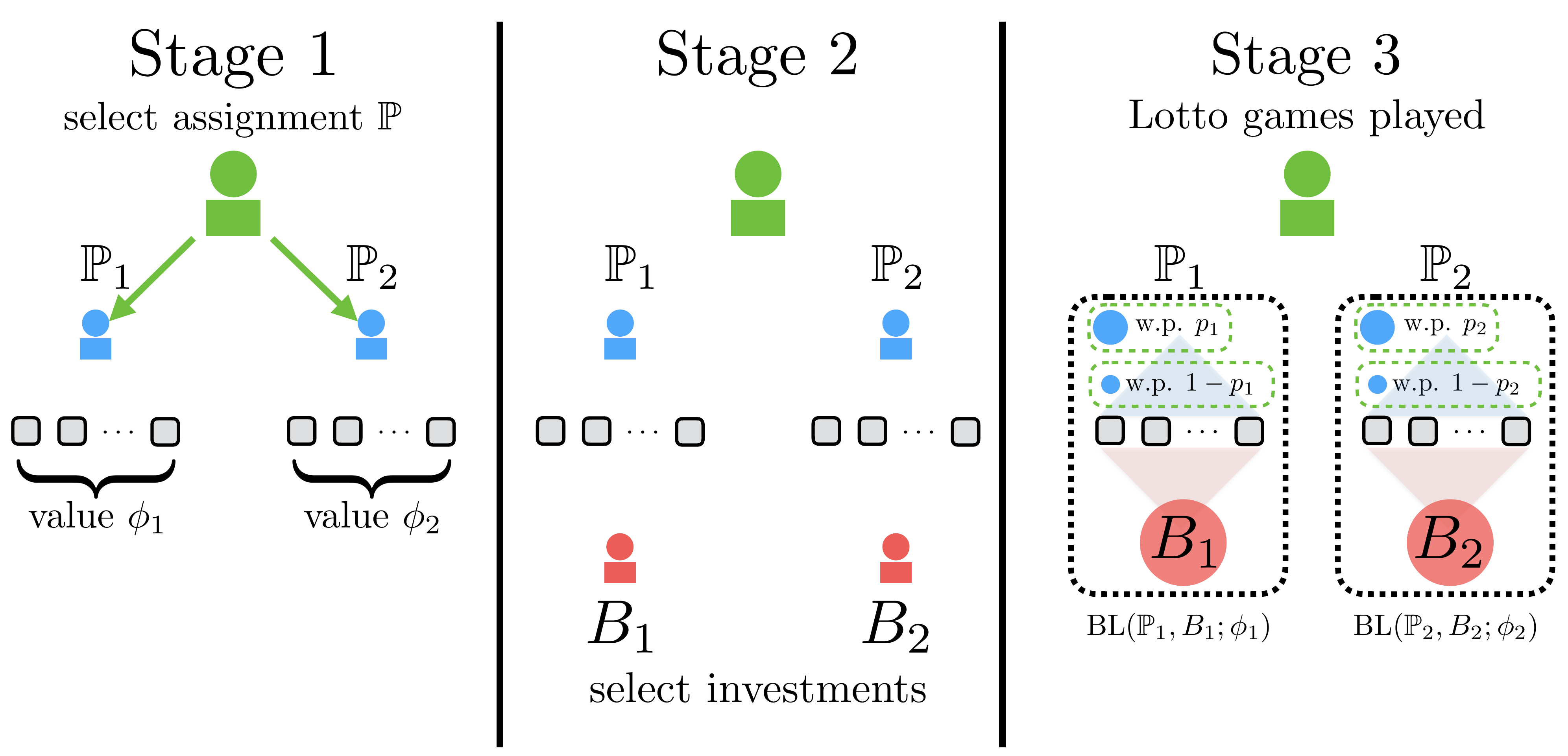}
    \caption{The three-stage commander assignment problem. In stage 1, the commander decides how to assign resources to two sub-colonels. The assignment strategy $\Prob$ is either randomized or deterministic. The commander pays a cost $c\cdot\E_\Prob[A_1+A_2]$ for the expected amount of assigned resources. Sub-colonel will use its assigned resources to compete over a set of battlefields of value $\phi_i$. In Stage 2, the two opponents observe the assignment strategy $\Prob$, but not the actual realizations (if randomized). Then, the opponents decide to invest in resource budgets $B_1,B_2$, for which they pay a cost $c_i B_i$. In Stage 3, two Bernoulli Lotto games are played between each sub-colonel and opponent using their endowed resources decided in Stages 1 and 2. The commander's payoff is the sum of sub-colonels' equilibrium payoffs in their respective games minus the cost to assign the resources from Stage 1 \eqref{eq:commander_U}. }
    \label{fig:CAP_stages}
\end{figure}

\noindent\textbf{Stage 2:} After observing the assignment strategy $\Prob$ from Stage 1, each opponent individually decides an amount of resources to invest in, ${B}_1 \leq \bar{B}_1$ and ${B}_2 \leq \bar{B}_2$. Each opponent $i$ will pay a cost $c_i {B}_i$, where $c_i > 0$ is the opponent's per-unit cost to invest in resources. It will use its invested resources to compete against sub-colonel $i$ on a set of battlefields with total value $\phi_i > 0$.

\vspace{1mm}

\noindent\textbf{Stage 3:} Two independent Bernoulli Lotto games are played simultaneously: $\mcg_1=\text{BL}(\Prob_1,{B}_1; \phi_1)$ between sub-colonel $\mcg_2=\mca_1$ and opponent $\mcb_1$, and $\text{BL}(\Prob_2,{B}_2; \phi_2)$ between sub-colonel $\mca_2$ and opponent $\mcb_2$. Here, $\phi_i > 0$ indicates the total value of the set of battlefields contested in $\mcg_i$. The final payoff that the commander obtains is given by the sub-colonels' cumulative equilibrium payoffs from $\mcg_1$ and $\mcg_2$ minus the costs of expenditure:
\begin{equation}\label{eq:commander_U}
    W(\Prob,\{{B}_i\}_{i=1,2}) \triangleq \phi_1\cdot\pi_\mca(\Prob_1,{B}_1) + \phi_2\cdot\pi_\mca(\Prob_2,{B}_2) - c\cdot\E_{\Prob}[A_1 + A_2]
\end{equation}
The final payoff that opponent $\mcb_i$ obtains is
\begin{equation}\label{eq:opponent_U}
    U_i(\Prob,{B}_i,c_i) \triangleq \phi_i\cdot\pi_\mcb(\Prob_i,{B}_i) - c_i {B}_i.
\end{equation}

We assume all parameters are common knowledge to all players. The above decision problem can be solved by backwards induction. Note that given decisions from the commander and opponents, the final payoffs in Stage 3 can be computed directly from Theorem \ref{thm:BNE_regions}. Thus, we first address the optimal investment decisions for the opponents in Stage 2, i.e. ones that maximize \eqref{eq:opponent_U} over $B_i \geq 0$. We will assume that when there are multiple maximizers of \eqref{eq:opponent_U}, player $\mcb_i$ chooses the smallest investment level among them. We then solve for the optimal assignment strategy for the commander in Stage 1, i.e. that maximizes \eqref{eq:commander_U} given that the opponents invest optimally. The full analysis of these arguments is provided in Appendices \ref{sec:deterministic_proofs} and \ref{sec:randomized_proofs}.

The above extensive-form game will be denoted as $\text{CAP}(\bar{A}_C,c,\{\bar{B}_i,c_i,\phi_i\}_{i=1,2})$, and we will denote $\text{CAP}_\text{d}(\bar{A}_C,c,\{\bar{B}_i,c_i,\phi_i\}_{i=1,2})$ as the extensive-form game where the commander is restricted to deterministic assignment strategies. We will consider two different settings on the parameters.
\begin{itemize}
	\item The \emph{fixed budget setting}. Here, there are zero costs associated with using resources, i.e. $c = c_i = 0$, and $\bar{A}_C$, $\bar{B}_i < \infty$. The commander assignment problem $\text{CAP}_\text{d}$ under this setting was first featured in \cite{Kovenock_2012}, where optimal deterministic assignments are provided.
	\item The \emph{per-unit cost setting}. Here, the commander and opponents have costs $c > 0$ and $c_i > 0$. There are no resource budget limits, i.e. $\bar{A}_C, \bar{B}_i = \infty$. We will then simply denote the space of feasible commander strategies as $\mcal{P}$. While the commander assignment problem under this setting has not been considered in the literature (to the best of our knowledge), formulations of two player Colonel Blotto games are commonly analyzed under similar linear cost models \cite{Kvasov_2007,Kovenock_handbook_2012,Kim_2017,Kovenock_2018}.
\end{itemize}
To the best of our knowledge, the commander assignment problem with randomized assignment strategies (in either setting) is novel to the literature. The primary goal in this section is to quantify the performance improvement (if any) that randomized assignments can yield over deterministic assignments for the commander.

\subsection{Results: the value of randomization}

Here, we present our results concerning the comparison of the commander's performance in $\text{CAP}$ and $\text{CAP}_\text{d}$. Denote $W^*$ as the optimal final payoff to the commander from $\text{CAP}$ (randomized assignments) and $W_\text{d}^*$ as its final optimal payoff from $\text{CAP}_\text{d}$ (deterministic assignments).

\begin{theorem}\label{thm:W_theorem}
    The following statements hold.
    \begin{itemize}
        \item Under the fixed budget setting, i.e. $c=c_i=0$ and $\bar{A}_C, \bar{B}_i < \infty$, it holds that $W^* = W_\text{d}^*$.
        \item Under the per-unit cost setting, i.e. $c,c_i > 0$ and $\bar{A}_C,\bar{B}_i = \infty$, it holds that $0 < W_\text{d}^* \leq W^* \leq 4\cdot W_\text{d}^*$. The equality $W^* = 4\cdot W_\text{d}^*$ holds if and only if $c > \max\{\frac{1}{2}\sqrt{\frac{c_j}{\phi_j}(c_1\phi_1+c_2\phi_2)},c_k(1+\sqrt{3}/2)\}$, where $j = \argmin{i=1,2} \frac{\phi_i}{2c_i}$ and $k = \argmax{i=1,2} c_i$.
    \end{itemize}
\end{theorem}
First, let us note that the characterization of $W^*$ -- the optimal final payoff the commander can obtain by randomizing assignments -- in either setting is novel to the literature. The first statement in Theorem \ref{thm:W_theorem} asserts that the commander cannot profitably exploit informational asymmetries by randomizing assignments when there are use-it or lose-it finite resource budgets (fixed budget setting) with zero marginal costs of expenditure. A complete characterization of $W_\text{d}^*$ under this setting is available in the literature \cite{Kovenock_2012}. For brevity, we will not include it here. The second statement of Theorem \ref{thm:W_theorem} asserts that the commander can attain up to a four-fold performance improvement over deterministic assignments when there are non-zero marginal costs involved. The four-fold improvement is attainable in a regime where the commander's per-unit cost is sufficiently high, i.e. when it is relatively expensive to assign  resources to the sub-colonels.

The derivation of the second statement in the above result is detailed in Appendices \ref{sec:deterministic_proofs} and \ref{sec:randomized_proofs}. In Appendix \ref{sec:deterministic_proofs}, we use backward induction arguments to completely characterize the optimal \emph{deterministic} assignment in $\text{CAP}_\text{d}$ and the corresponding final payoff $W_\text{d}^*$. In Appendix \ref{sec:randomized_proofs}, we use backward induction arguments to completely characterize the optimal \emph{randomized} assignment in $\text{CAP}$ and the corresponding final payoff $W^*$.

\section{Comparative statics and numerical study}\label{sec:sims}

In this section, we highlight several implications of our main results using analytical and numerical methods. We revisit the Bernoulli Lotto game $\text{BL}(\mbb{P}_\mca,B)$ from Section \ref{sec:BL} and compare it to a corresponding benchmark complete information General Lotto game. In particular, we analyze player $\mca$'s equilibrium payoff in two comparative scenarios: 1) the Bernoulli Lotto game $\text{BL}(\mbb{P}_\mca,B)$ where $\mbb{P}_\mca = (A^h,A^\ell,p)$ and 2) the complete information Lotto game $\text{GL}(\bar{A},B)$ in which $\mca$ utilizes the average budget $\bar{A} = pA^h + (1-p)A^\ell$ from scenario 1. It is an immediate consequence\footnote{A third scenario in which the budget type is fully revealed to player $\mcb$ may also be considered. Compared to this scenario, it must be the case that player $\mca$ performs better in Scenario 1. We leave a full investigation of such comparisons to future work.} that player $\mca$ can do no better in scenario 1 than in scenario 2.

\begin{proposition}\label{prop:CI_better}
	For any $\bar{A} > 0$ and $B > 0$, it holds that $\pi_\mca(\Prob_\mca,B) \leq \pi_\mca^\text{CI}(\bar{A},B)$ for any $\Prob_\mca = (A^h,A^\ell,p)$ satisfying $pA^h + (1-p)A^\ell = \bar{A}$, where $\pi_\mca^\text{CI}(\bar{A},B)$ is the equilibrium payoff of the complete information General Lotto game $\text{GL}(\bar{A},B)$.
\end{proposition}
\begin{proof}
	Let us denote $F^*_\mca \in \F(\bar{A})$ as the equilibrium strategy for player $\mca$ in the complete information General Lotto game (without costs), which ensures a payoff $\pi_\mca^\text{CI}$. In the BL game, player $\mca$ selects a pair of strategies $\{F_\mca^h,F_\mca^\ell\} \in \F(A^h)\times \F(A^\ell)$. Player $\mca$ effectively ex-ante competes with a single strategy $\bar{F}_\mca = pF_\mca^h + (1-p)F_\mca^\ell$ that belongs to $\F(\bar{A})$. This follows directly from the linearity of expectations in both \eqref{eq:LC} and \eqref{eq:Lotto_payoff}. However, the set $\{\bar{F}_\mca \in \F(\bar{A}): pF_\mca^h + (1-p)F_\mca^\ell \}$ is a subset of $\F(\bar{A})$, which may not contain $F_\mca^*$. 
\end{proof}

Note that the explicit characterization of equilibrium payoffs, given in Theorem \ref{thm:BNE_regions}, was not needed to prove the above statement. In essence, the fact that player $\mca$ has two use-it or lose-it budget types \emph{constrains} its feasible strategy set in comparison to its strategy set in the corresponding complete information game.

We are thus interested in characterizing the extent to which player $\mca$'s equilibrium payoff degrades when its budget is randomized, in comparison to the benchmark complete information game. To illustrate, let us first consider an example where the average budget $\bar{A}$ and the probability $p$ are fixed. The feasible budget pairs for player $\mca$ are then parameterized by $A^h \in [\bar{A},\frac{\bar{A}}{p}]$ with $A^\ell = \frac{\bar{A} - pA^h}{1-p}$. The equilibrium payoffs to player $\mca$ under this parameterization with $\bar{A} = 2$ is shown in Figure \ref{fig:slice}. Notably, as the  separation between $A^h$ and $A^\ell$ increases, the equilibrium payoff to $\mca$ worsens.

\begin{figure}[t]
    \centering
    \includegraphics[scale=0.15]{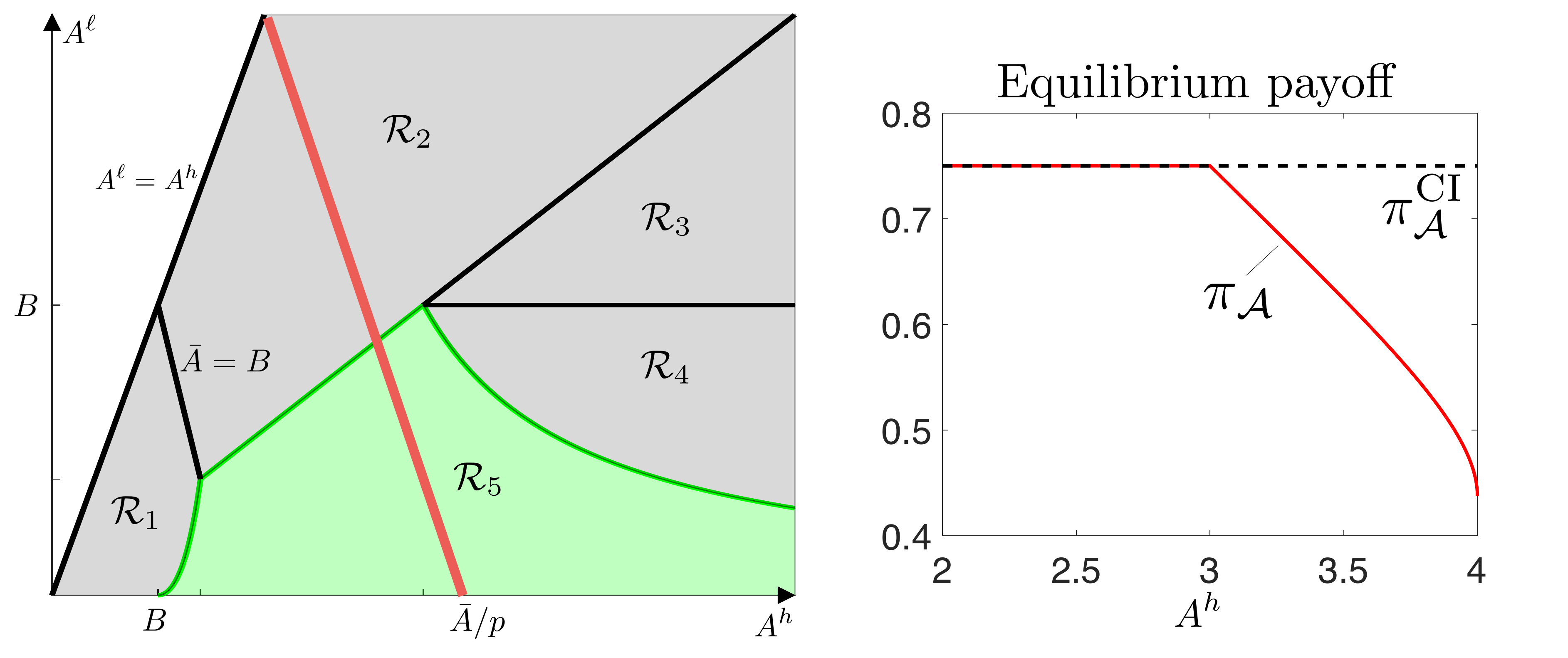}
    \caption{(Left) Parameter region diagram for a fixed $p$, where the red line highlights all possible budget pairs $(A^h,A^\ell)$ that satisfy a fixed average budget $\bar{A}$. (Right) Player $\mca$'s equilibrium payoff in the Bernoulli Lotto game (red line) for all budget pairs that satisfy a fixed average budget $\bar{A} = 2$ with fixed parameter $p = 0.5$ (indicated by the thick red line in left panel). This is parameterized by  $A^h \in [\bar{A},\frac{\bar{A}}{p}]$ with $A^\ell = \frac{\bar{A} - pA^h}{1-p}$. For instance, when $A^h = 3.5$, the distribution $\Prob_\mca = (3.5,0.5,0.5)$. Observe that randomized budgets worsen player $\mca$'s payoff over its performance in the benchmark complete information game (dashed black line). The extreme distribution $\Prob_\mca =(4,0,0.5)$ only achieves $58\%$ of the payoff in the benchmark game.}
    \label{fig:slice}
\end{figure}

In order to generalize these observations, we define the following quantity for any distribution $\Prob_\mca = (A^h,A^\ell,p)$:
\begin{equation}
    \zeta(\Prob_\mca;B) \triangleq \frac{\pi_\mca(\Prob_\mca,B)}{\pi_\mca^\text{CI}(\bar{A},B)}.
\end{equation}
This is the fraction of payoff that $\Prob_\mca$ achieves in the BL game relative to the benchmark complete information game, which is between 0 and 1 by Proposition \ref{prop:CI_better}. Now, suppose the probability $p \in [0,1]$ of a high budget is fixed. We wish to characterize the worst-case degradation over all possible budget pairs $(A^h,A^\ell)$, subject to a lower bound on the separation: $\alpha \leq \frac{A^\ell}{A^h} \leq 1$ for some $\alpha \in (0,1]$. This constraint imposes a minimum relative spending or budget requirement. For instance, defense budget spending typically is not extremely high one year and near zero the next year. The value we are thus interested in quantifying is given by the \emph{degradation ratio}:
\begin{equation}\label{eq:degradation_p}
    \zeta^*(p;\alpha,B) \triangleq  \min_{\substack{ (A^h,A^\ell) \\ \alpha \leq \frac{A^\ell}{A^h} \leq 1}} \zeta((A^h,A^\ell,p);B).
\end{equation}

Below, we explicitly derive the degradation ratio in player $\mca$'s equilibrium payoff for a fixed $p \in [0,1]$.


\begin{proposition}\label{prop:zeta_p}
    For any fixed $p \in [0,1]$ and $B > 0$,
    \begin{equation}\label{eq:zeta_p}
        \zeta^*(p;\alpha,B) = 
        \begin{cases}
            p, &\text{if } \alpha = 0 \\
            \left(1 - \frac{B}{2T^*(p+\alpha(1-p))} \right)^{-1} \cdot \left(p + (1-p)\frac{\alpha T^*}{2B} \right), &\text{if } \alpha < \frac{1-p}{2-p} \\
            1, &\text{if } \frac{1-p}{2-p} \leq \alpha \leq 1 \\
        \end{cases}
    \end{equation}
    where
    \begin{equation}\label{eq:Tstar}
        T^* \triangleq \frac{B\left(\alpha(1-p) + \sqrt{\alpha^2(1-p)^2 + 4p\alpha(1-p)(p+\alpha(1-p))} \right)}{2\alpha(1-p)(p+\alpha(1-p))}.
    \end{equation}
\end{proposition}

As the separation parameter $\alpha$ increases, meaning that the low and high budgets are closer in value, the degradation ratio increases. Above the threshold $\alpha \geq \frac{1-p}{2-p}$, the equilibrium payoff from the BL game coincides with the benchmark complete information game. We find that the budget pair satisfying $\frac{A^\ell}{A^h} = \alpha$ minimizes \eqref{eq:degradation_p}, i.e. the further apart the high and low budget are, the worse performance $\mca$ attains.

\begin{proof}
    We first establish that $\pi_\mca$ is a decreasing function of $A^h$ along any line segment of the form $A^\ell = \frac{\bar{A} - pA^h}{1-p}$, where $A^h \in [\bar{A},\frac{\bar{A}}{p}]$. This is the line segment that encompasses all pairs $(A^h,A^\ell)$ satisfying an expected budget $\bar{A} > 0$. We first note that $\pi_\mca = \pi_\mca^{\text{CI}}$ in the region $\mcal{R}_1 \cup \mcal{R}_2$. Moreover, $\pi_\mca$ is a continuous function in $A^h,A^\ell \in \mcal{R}$ and hence $\pi_\mca = \pi_\mca^{\text{CI}}$ for any $(A^h,A^\ell)$ on the border of $\mcal{R}_5$ and $\mcal{R}_1$, $\mcal{R}_2$. Let us first consider the case $\bar{A} < 1 + \frac{1}{1-p}$. The line segment begins in $\mcal{R}_1$ or $\mcal{R}_2$, where $\pi_\mca$ is constant. It then enters $\mcal{R}_5$. Here, we need to show that the last entry of \eqref{eq:piA} is decreasing in $A^h$ along the segment. The partial derivative with respect to $A^h$, keeping $\bar{A}$ constant, is $-\frac{B}{(A^h)^2}\left(\bar{A} + \sqrt{\frac{\bar{A}}{\bar{A}-pA^h}}\left(\frac{pB}{2}+\bar{A} - pA^h \right) \right) + \frac{B^2}{(A^h)^3}\left(1+\sqrt{\frac{\bar{A}}{\bar{A}-pA^h}} \right)\left(\left(\bar{A}-pA^h + \sqrt{\bar{A}(\bar{A}-pA^h)} \right) - \frac{B}{2} \right)$.   Note the first term is negative, and the second term is negative if $p^2(A^h)^2 + p(B-\bar{A})A^h + \frac{B^2}{4} > 0$. This indeed is the case -- observe this is a quadratic function that is positive for all $A^h \geq 0$. Therefore, $\pi_A$ is decreasing in $A^h$ along the line segment in region $\mcal{R}_5$.
    
    Now, consider the case $\bar{A} \geq B(1 + \frac{1}{1-p})$. The line segment begins in $\mcal{R}_2$, then enters $\mcal{R}_3$, $\mcal{R}_4$, and $\mcal{R}_5$. We can easily verify $\pi_\mca(G)$ is decreasing in regions $\mcal{R}_3$ and $\mcal{R}_4$ along the segment. For $\mcal{R}_3$, the partial derivative of \eqref{eq:piA} along the segment is $-\frac{B^2(1-p)^2}{2(\bar{A}-pA^h)^2} < 0$. For $\mcal{R}_4$, the partial derivative of \eqref{eq:piA} along the segment is $-p/2 < 0$.
    
    Consequently, the approach to calculate \eqref{eq:degradation_p} is to search $\zeta$ over pairs $(A^h,\alpha A^h)$, i.e. along the ray that meets the separation constraint. If $\alpha \geq \frac{1-p}{2-p}$, we can immediately deduce that $\zeta^*(p;\alpha,B) = 1$, since the ray is entirely contained in $\mcal{R}_1$ and $\mcal{R}_2$. So, suppose $\alpha < \frac{1-p}{2-p}$. Along the ray parameterized by $A^h > 0$, we obtain
    \begin{equation}
        \zeta(A^h) = 
        \begin{cases}
            1, &\text{if } A^h \in [0,T_1] \\
            \frac{2B}{A^h(p+\alpha(1-p))}\cdot \pi_\mca((A^h,\alpha A^h,p),B), &\text{if } A^h \in (T_1,T_2] \\
            \left(1 - \frac{B}{2A^h(p+\alpha(1-p))} \right)^{-1} \cdot \pi_\mca((A^h,\alpha A^h,p),B), &\text{if } A^h \in (T_2,T_3] \\
            \left(1 - \frac{B}{2A^h(p+\alpha(1-p))} \right)^{-1} \cdot(p+\alpha(1-p)\frac{A^h}{2B}), &\text{if } A^h \in (T_3,T_4] \\
            \left(1 - \frac{B}{2A^h(p+\alpha(1-p))} \right)^{-1}(p + (1-p)(1 - \frac{B}{2\alpha A^h}) ), &\text{if } A^h \in (T_4,\infty]
        \end{cases}
    \end{equation}
    where $T_1 \triangleq B(1 + \sqrt{\frac{\alpha(1-p)}{p + \alpha(1-p)}})$, $T_2 = \frac{B}{p+\alpha(1-p)}$, $T_3 \triangleq B(1+\sqrt{1+\frac{p}{\alpha(1-p)}})$, and $T_4 \triangleq B/\alpha$. For $A^h \in [0,T_1]$, the ray is contained in $\mcal{R}_1$. For $A^h \in (T_1,T_3]$, the ray is contained in $\mcal{R}_5$, and it can be shown that the expressions for $\zeta$ are strictly decreasing in $A^h$. For $A^h \in (T_3,T_4]$, the ray is contained in $\mcal{R}_4$, and one finds a unique minimum point in this interval at $A^h = T^*$ \eqref{eq:Tstar}. For $A^h \in (T_4,\infty)$, the ray is contained in $\mcal{R}_3$, and we can  show that $\zeta$ is strictly increasing. Therefore, $\zeta^*(p;\alpha,B) = \zeta(T^*)$.
\end{proof}

\begin{figure}[t]
    \centering
    \includegraphics[scale=0.15]{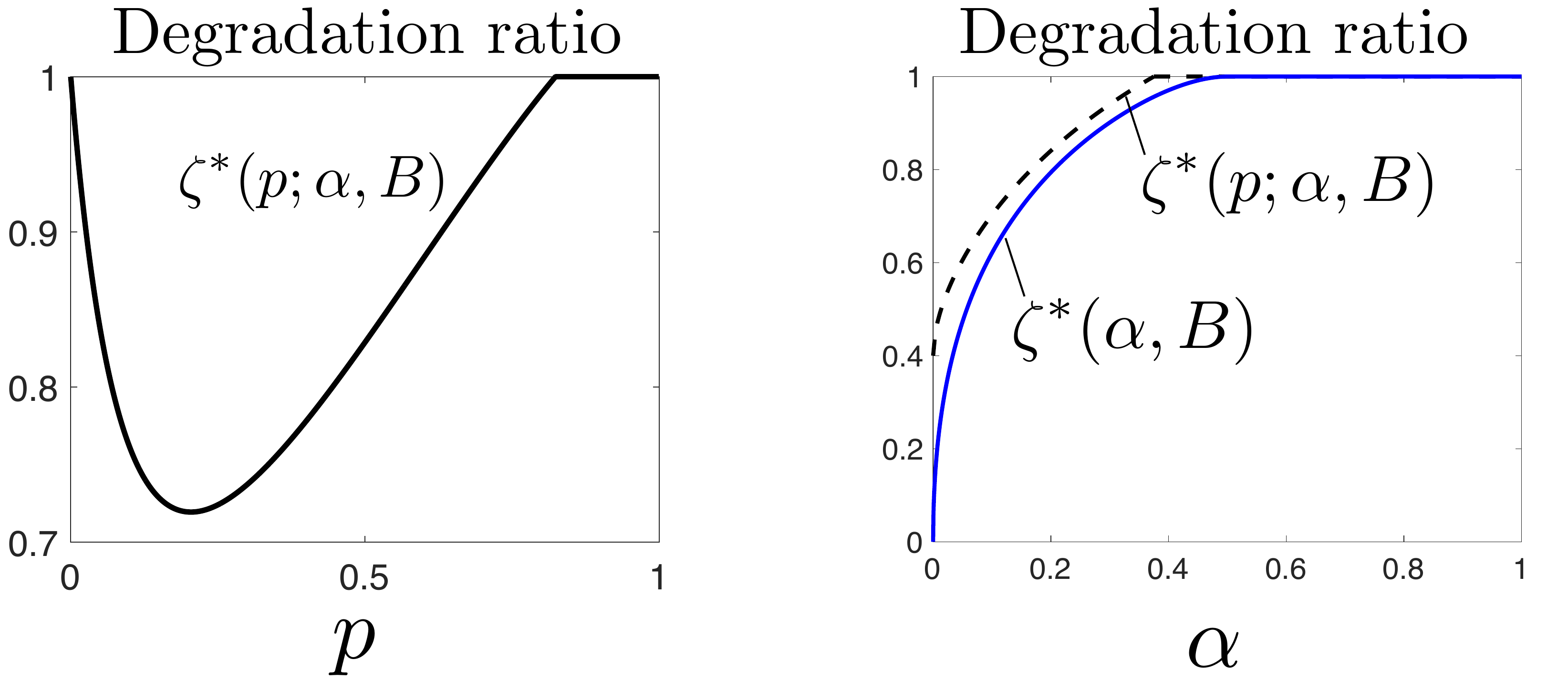}
    \caption{(Left) Plot of the  degradation ratio $\zeta^*(p;\alpha,B)$ \eqref{eq:zeta_p}, characterized explicitly in Proposition \ref{prop:zeta_p}. As a function of $p$, we observe $\zeta^*(p;\alpha,B)$ has a unique global minimizer. In this plot, $\alpha = 0.15$. (Right) The worst-case degradation ratio $\zeta^*(\alpha,B)$ \eqref{eq:degradation_Prob} is computed numerically by performing a gradient descent algorithm on $\zeta^*(p;\alpha,B)$. We note that $\zeta^*(\alpha,B) = 1$ for $\alpha \geq 1/2$. As a comparison, $\zeta^*(p;\alpha,B)$ is plotted for $p = 0.4$. In all simulations, we set $B=1$. }
    \label{fig:degradation_plots}
\end{figure}

One may also be interested in the \emph{worst-case} degradation ratio, in which the fraction $\zeta(\Prob_\mca;B)$ is minimized over all possible Bernoulli distributions $\Prob_\mca = (A^h,A^\ell,p)$ subject to the separation constraint $\alpha \leq \frac{A^\ell}{A^h} \leq 1$. Such a characterization requires solving the optimization problem
\begin{equation}\label{eq:degradation_Prob}
    \zeta^*(\alpha,B) \triangleq \min_{p \in [0,1]} \zeta^*(p;\alpha,B).
\end{equation}
where $\zeta^*(p;\alpha,B)$ is given in \eqref{eq:zeta_p}. To calculate \eqref{eq:degradation_Prob}, we leverage numerical methods (gradient descent) to accurately find its value. The simulation results are depicted in Figure \ref{fig:degradation_plots}. We observe that when there is no separation constraint $(\alpha = 0)$, the degradation is arbitrarily bad, i.e. $\zeta^*(0,B) = 0$. Here, the worst-case distribution is one that sets $A^\ell = 0$ and $p=0$. The worst-case ratio $\zeta^*(\alpha,B)$ is strictly increasing in $\alpha \in (0,\frac{1}{2})$, and there is no degradation ($\zeta^*(\alpha,B) = 1$) when the high budget is no more than twice the low budget, i.e.  $\alpha \in [\frac{1}{2},1]$.

\section{Conclusion}\label{sec:conclusion}

This paper considers a class of one-sided incomplete information General Lotto games, where one of the player's resource budget is assigned randomly according to a Bernoulli distribution, while the opponent's endowment is common knowledge. We characterize all Bayes-Nash equilibria in this class of games. Randomized budgets do not provide benefits in comparison to a corresponding benchmark complete information setting. Our equilibrium characterizations allow us to determine how a high-level commander could exploit probabilistic resource assignments to two sub-colonels that engage with two respective opponents in separate and simultaneous General Lotto games. Interestingly, the optimal randomized assignment can improve the commander's payoff four-fold in comparison to the optimal deterministic assignment, in settings with a per-unit cost for deployment.

There are several interesting directions for future research. One can consider extensions in which there are more than two budget types, or where both competitors have randomized budgets. The equilibria derived in these games may also serve as equilibria to a class of corresponding all-pay auctions with incomplete information. It would also be of interest to extend the assignment problems to account for more than two sub-colonels.

\section*{Acknowledgments}
This work is supported by UCOP Grant LFR-18-548175, ONR grant \#N00014-20-1-2359, AFOSR grants \#FA9550-20-1-0054 and \#FA9550-21-1-0203, and the Army Research Lab through the ARL DCIST CRA \#W911NF-17-2-0181.

\begin{appendices}

\section{Equilibrium characterizations of Bernoulli Lotto games}\label{sec:APA}


This section is devoted to the proof of Theorem \ref{thm:BNE_regions} -- the characterization of equilibrium payoffs in the Bernoulli General Lotto game. To do so, we derive the equilibrium strategies for player $\mca$ and $B$ in any game instance $\text{BL}(\Prob_\mca,B,0)$. First, we establish a connection between the necessary conditions for equilibrium in two-player all-pay auctions with asymmetric information and the BL game (Section \ref{sec:connection_method}). We then leverage equilibrium solutions to such all-pay auctions (Section \ref{sec:APA_IUI}), provided by the work of Siegel \cite{Siegel_2014}, to formulate a system of non-linear equations in the Lagrange multipliers $\bs{\lambda}$ associated with the players' Lotto expected budget constraints \eqref{eq:LC}. This system of equations is stated in Section \ref{sec:SB_equil}, and can take three different forms, which correspond to  disjoint regions in the multiplier space. Hence, solutions to this system are not only algebraic, but also case-dependent.   

We completely characterize solutions to these equations in Proposition \ref{prop:SB_equil}, finding they exist only for a subset of BL games we are interested in (region $\mcal{R}_5$). Indeed, Siegel's algorithm  \cite{Siegel_2014} constructs equilibrium strategies to the auctions when certain monotonicity conditions are met -- players' types are required to be somewhat correlated to their valuations. While these informational requirements hold in the BL games we are interested in, we find the structure of the auction strategies cannot accommodate all possible combinations of budget parameters $A^h,A^\ell,B$, limiting the applicability of \cite{Siegel_2014} to the sub-region $\mcal{R}_5$. Nonetheless, we prove that solutions to the system of equations, when they exist, correspond to equilibria of the BL game (Proposition \ref{prop:SB_equil}).

We then identify the remaining regions $\mcal{R}_i$, $i=1,\ldots,4$ to have distinct equilibrium structures that cannot be calculated from the aforementioned system of equations. For these equilibria, we combine features of the equilibrium strategies from complete information Lotto games with those that were computed through the system of equations. These details are given in the Section \ref{sec:other_regions}, thus completing the proof of Theorem \ref{thm:BNE_regions}. We first review all-pay auctions with asymmetric information, as studied by Siegel in \cite{Siegel_2014}.

\subsection{All-pay auctions with asymmetric information and valuations}
In an all-pay auction, two bidders ($A$ and $B$) compete over a single item. Before bids are submitted, player $\mca$ privately observes one of two possible types: $h$ (high type) with a probability $p$ and $\ell$ (low type) with $1-p$. Here, it is assumed $p$ is common knowledge and $p > 0$. player $\mcb$ always observes the same type $t_\mcb$. Hence, there are two possible type profiles\footnote{Siegel's model is general, allowing an arbitrary, finite number of types for each player. We review the model here with two types, since it pertains to our Bernoulli General Lotto games.}, corresponding to whether $\mca$ observes $h$ or $\ell$. In type $t\in\{h,\ell\}$, the players' valuations for the item are $v_{\mca,t}$ and $v_{\mcb,t}$ for player $\mca$ and $\mcb$, respectively. 

A (pure) strategy for player $\mca$ is a pair $\bs{x}_\mca = (x_h,x_\ell) \in \mathbb{R}_+^2$, where  $x_t$ is its bid for the item contingent on receiving type $t$. A strategy for player $\mcb$ is a non-negative bid $x_\mcb \geq 0$. The resulting payoffs are given by
\begin{equation}
    \begin{aligned}
        p\left( v_{\mca,h}\cdot W_A(x_h,x_\mcb) - x_h \right) + (1-p) \left( v_{\mca,\ell}\cdot W_A(x_\ell,x_\mcb) - x_\ell \right) \quad \text{(Player $\mca$)} \\
        p\left( v_{\mcb,h}\cdot W_B(x_\mcb,x_h) - x_\mcb \right) + (1-p) \left( v_{\mcb,\ell}\cdot W_B(x_\mcb,x_\ell) - x_\mcb \right) \quad \text{(Player $\mcb$)}
    \end{aligned}
\end{equation}
where $W_\ell$, $\ell \in \{A,B\}$ is defined as\footnote{The choice of tie-breaking rule can be made arbitrarily, as it will not ultimately affect equilibrium payoffs. Here, we assume ties favor player $\mcb$.}
\begin{equation}\label{eq:W}
	W_\ell(a,b) =
	\begin{cases}
		1, &\text{if } a > b \text{ or } a = b \text{ and } \ell = B\\
		0, &\text{otherwise}
	\end{cases}
\end{equation}
A mixed strategy for player $\mca$ is a pair  $\{F_\mca^{t}\}_{t=h,\ell}$, where $F_\mca^{t}$ is a univariate probability distribution on $\mathbb{R}_+$. A mixed strategy for player $\mcb$ is a single univariate distribution $F_\mcb$ on $\mathbb{R}_+$. The payoffs are calculated as the expected payoffs with respect to the distribution $p$ and the mixed strategies. We refer to this auction as $\text{APA}(v_\mca,v_\mcb,p)$.

Note that we have adapted this model to the information structure of the Bernoulli Lotto game. In general, the auction model allows for arbitrary, finite type spaces $\mcal{T}_\mca$ and $\mcal{T}_\mcb$ with a joint probability distribution $\bs{p}$ on the type profiles \cite{Siegel_2014}. Siegel shows equilibria can be computed algorithmically, provided a certain monotonicity condition (known as the KMS condition in the auction literature \cite{Rentschler_2016}) is met for some ordering of the players' type spaces. For brevity, we do not detail these conditions here. Indeed, with the information structure that we consider, the monotonicity condition always holds.

\subsection{Connection between APA and Bernoulli General Lotto games}\label{sec:connection_method}

We now present some informal intuition that suggests a connection between the equilibria of APA and $\text{BL}(\Prob_\mca,B,0)$, the Bernoulli Lotto game with zero cost for player $\mcb$. These insights are analogous to those drawn between two-player all-pay auctions with complete information and the Colonel Blotto and General Lotto games \cite{Baye_1996,Roberson_2006,Kovenock_2020}.

Consider a game instance $\text{BL}(\Prob_\mca,B,0)$, which we will often refer to as the tuple of parameters $G = (A^h,A^\ell,p,B)$. For ease of exposition, we will assume that the sum of battlefield values is normalized to one, i.e. $\|\bs{v}\| = 1$. The Lotto budget constraint \eqref{eq:LC} must hold for the strategies associated in each type. player $\mca$'s ex-interim constrained optimization, given type $t\in\{h,\ell\}$ is realized, can be written as
\begin{equation}\label{eq:lagrangian}
    \max_{ \{F_{\mca,j}^{t}\}_{j\in[n] }} 
    \sum_{j\in[n]} \int_0^\infty \left[  v_j F_{B,j}(x_{\mca,j}) - \lambda_i x_{\mca,j} \right] \,dF_{\mca,j}^{t} + \lambda^t A^t
\end{equation}
where $\lambda^t$ is the multiplier on player $\mca$'s expected budget constraint for type $t\in\{h,\ell\}$, and we denote $p^h = p$ and $p^\ell = 1-p$. Player B's constrained optimization is written as
\begin{equation}
    \max_{ \{F_{B,j}\}_{j\in[n]}} \sum_{j\in[n]}
    \sum_{t=h,\ell} p^t \int_0^\infty \left[  v_jF_{\mca,j}^{t}(x_{B,j}) - \lambda_\mcb x_{B,j} \right] \,dF_{B,j} + \lambda_\mcb B.
\end{equation}
where $\lambda_\mcb$ is the multiplier on player $\mcb$'s budget. The necessary first-order conditions for equilibrium are
\begin{equation}\label{eq:lotto_connection}
    \begin{aligned}
        &\frac{d}{dx_{B,j}}\left[ \sum_{t=h,\ell} p^t \left(v_j F_{\mca,j}^{t}(x_{B,j}) - \lambda_\mcb x_{B,j} \right) \right] = 0 \\
        &\frac{d}{dx_{\mca,j}}\left[ v_j F_{B}(x_{\mca,j}) - \lambda^t x_{\mca,j}  \right] = 0, \quad t = h,\ell
    \end{aligned}
\end{equation}
for each $j \in [n]$. Dividing by the associated (positive) multiplier in each condition, this coincides with the necessary first-order conditions for equilibrium of $n$ independent two-player all-pay auctions with incomplete information for which the item valuation in auction $j$ for player $\mca$ in type $t\in\{h,\ell\}$ is $v_{A,i} = \frac{v_j}{\lambda_i}$, and the valuation for player $\mcb$ when player $\mca$'s type is $i$ is $v_{B,i} = \frac{v_j}{\lambda_\mcb}$.

The equilibria to each of the $n$ APA games can be computed using Siegel's algorithm, as long as the KMS condition is satisfied\footnote{As discussed earlier, the algorithm of \cite{Rentschler_2016} can handle APA when KMS is not met. We do not detail this algorithm here, however, because of its complexity and because Siegel's algorithm suffices for the problems of interest in this paper.}. However, since we do not know the actual ranking of player $\mca$'s types, i.e. whether $\lambda^h \leq \lambda^\ell$ or vice versa, we must proceed by first imposing such a ranking. For the sake of demonstration, let us suppose $\lambda^h \leq \lambda^\ell$, so that type $h$ is ``higher" than type $\ell$, for instance. This allows us to proceed with Siegel's algorithm, which generates bidding distributions $\{F_{\mca,j}^{t,\bs{\lambda}}\}_{t=h,\ell}$ and $F_{B,j}^{\bs{\lambda}}$  for each $j\in[n]$. Here, the superscript $\bs{\lambda}$ indicates the expressions are in terms of the (still) unknown multipliers. These distributions must be consistent with the Lotto expected budget constraints \eqref{eq:LC}, yielding a system of three equations in $\bs{\lambda} = (\lambda^h,\lambda^\ell,\lambda_\mcb)$.

\begin{equation}\label{eq:SOE}
\begin{aligned}
    &\sum_{j\in[n]} \mathbb{E}_{x_{\mca,j}\sim F_{\mca,j}^{t,\bs{\lambda}}}[x_{\mca,j}] = A^t, \quad t=h,\ell \\
    &\sum_{j\in[n]} \mathbb{E}_{x_{B,j}\sim F_{B}^{\bs{\lambda}}}[x_{B,j}] = B \\
    &\text{such that } 0 < \lambda^h \leq \lambda^\ell \\
\end{aligned}
\end{equation}
Hence, we seek to find a solution to \eqref{eq:SOE}. In the next sections, we detail the distributions $\{F_\mca^{t,\bs{\lambda}}\}_{t=h,\ell}$ and $F_\mcb^{\bs{\lambda}}$ constructed from Siegel's algorithm and apply them to \eqref{eq:SOE} to explicitly derive the system of equations. 

\subsection{Equilibrium strategies of APA}\label{sec:APA_IUI}

In the following, we summarize the resulting equilibria from applying Siegel's algorithm to the APA setup of Section \ref{sec:APA_IUI}.

Define 
\begin{equation}\label{eq:kbar}
	\bar{k} :=
	\begin{cases}
	    1, &\text{if } p\frac{v_{\mcb,h}}{v_{\mca,h}} \geq 1 \\
	    2, &\text{if } p\frac{v_{\mcb,h}}{v_{\mca,h}} + (1-p) \frac{v_{\mcb,\ell}}{v_{\mca,\ell}} \geq 1 \text{ and } p\frac{v_{\mcb,h}}{v_{\mca,h}} < 1 \\
	    3, &\text{if } p\frac{v_{\mcb,h}}{v_{\mca,h}} + (1-p) \frac{v_{\mcb,\ell}}{v_{\mca,\ell}} < 1
	\end{cases}
\end{equation}

In brief, $\bar{k}$ is the iteration at which Siegel's algorithm terminates. Denoting $p = p$ and $p_2 = 1-p$, define
%
%
\begin{equation}\label{eq:L}
	\begin{aligned}
	    L^h = 
	    \begin{cases}
	        v_{\mca,h}, &\text{if } \kb = 1 \\
	        p v_{\mcb,h}, &\text{if } \kb \in \{2,3\} \\
	    \end{cases}, \quad 
	    L^\ell = 
	    \begin{cases}
	        0, &\text{if } \kb = 1 \\
	        v_{\mca,\ell}\left( 1 - p\frac{v_{\mcb,h}}{v_{\mca,h}} \right), &\text{if } \kb = 2 \\
	        (1-p)v_{\mcb,\ell}, &\text{if } \kb = 3 \\
	    \end{cases}
	\end{aligned}
\end{equation}
The $L_k$ are lengths of intervals for which the equilibrium marginals have support. Below, we provide expressions for the equilibrium strategies that result from applying Algorithm 1.

\begin{lemma}\label{lem:APA_IUI}
	The equilibrium mixed strategies for APA are given as follows\footnote{To simplify exposition and notation where convenient, we sometimes explicitly write CDFs as a mixture of uniform and point mass distributions. Here, 
	we denote $\text{Unif}(a,b) := \bs{1}(x\geq a)\min\{\frac{x}{b-a},1\}$ as the CDF of the uniform distribution on $(a,b)$ and $\bs{\delta}_0 := \bs{1}(x\geq 0)$ the CDF of a point mass centered at zero.}:
	\begin{equation}\label{eq:A_marginals}
		\begin{aligned}
			\text{If } \kb = 1: \quad
			&F_\mca^{h} = \left(1 - \frac{L^h}{pv_{\mcb,h}} \right)\bs{\delta}_0 + \frac{L^h}{pv_{\mcb,h}}\text{\emph{Unif}}(0,L^h), \quad F_\mca^{\ell} = \bs{\delta}_0 \\
			&F_\mcb = \text{\emph{Unif}}(0,L^h) \\
			\text{If } \kb = 2: \quad
			&F_\mca^{h} = \text{\emph{Unif}}(L^\ell,L^\ell + L^h) \\
			&F_\mca^{\ell} = \left(1 - \frac{L^\ell}{(1-p) v_{\mcb,\ell}} \right)\bs{\delta}_0 + \frac{L^\ell}{(1-p) v_{\mcb,\ell}}\text{\emph{Unif}}(0,L^\ell) \\
			&F_\mcb = \frac{L^\ell}{v_{\mca,\ell}}\text{\emph{Unif}}(0,L^\ell) + \frac{L^h}{v_{\mca,h}}\text{\emph{Unif}}(L^\ell,L^\ell + L^h) \\
			\text{If } \kb = 3: \quad
			&F_\mca^{h} = \text{\emph{Unif}}(L^\ell,L^\ell + L^h), \quad F_\mca^{\ell} = \text{\emph{Unif}}(0,L^\ell) \\
			&F_\mcb = \left(1 - \sum_{i=1}^2 \frac{L_i}{v_{A,i}} \right)\bs{\delta}_0 + \frac{L^\ell}{v_{\mca,\ell}}\text{\emph{Unif}}(0,L^\ell) + \frac{L^h}{v_{\mca,h}}\text{\emph{Unif}}(L^\ell,L^\ell + L^h)
		\end{aligned}
	\end{equation}
\end{lemma}
In summary, the marginals for player $\mca$ are uniform distributions with shifted supports, and player $\mcb$'s marginal is a piece-wise uniform distribution.

\subsection{Equilibria in the $\mcal{R}_5$ region}\label{sec:SB_equil}

We are now ready to apply the methods outlined in Sections \ref{sec:connection_method} and \ref{sec:APA_IUI} to explicitly state the system of equations \eqref{eq:SOE}. We then completely characterize the solutions to these equations in Proposition \ref{prop:SB_equil}. In doing so, we identify the subset of game instances for which solutions to \eqref{eq:SOE} exist. Recall this subset was identified as the $\mcal{R}_5$ region in Theorem \ref{thm:BNE_regions} (Figure \ref{fig:regions}). We also prove that such solutions and their associated strategies (constructed from Siegel's algorithm) constitute equilibria of the BL game. This serves as the proof of Theorem \ref{thm:BNE_regions} in the $\mcal{R}_5$ region.

The item valuations in one of the $n$ APA games are given by $v_{A,t} = v_j/\lambda_i > 0$ for player $\mca$ and $v_{B,t} = v_j/\lambda_\mcb > 0$ for player $\mcb$, in type $t\in \{h,\ell\}$. Since the high budget $A^h$ is associated with type $t_1$, we naturally impose the ranking $\lambda^h \leq \lambda^\ell$. Indeed, the distribution $F_\mca^{h}$ \eqref{eq:A_marginals} will have a higher expected budget expenditure under this ranking of types. 

The value of $\kb$ is not known a priori, as it now depends on the multipliers. The values it can take, $\kb \in \{1,2,3\}$, correspond to transformed multipliers $\bs{\sigma} = (\sigma^h,\sigma^\ell)$, with $\sigma^t := \frac{\lambda^t}{\lambda_\mcb} > 0$ for $t \in \{h,\ell\}$, lying in three disjoint regions of $\mathbb{R}_+^2$. These regions result directly from \eqref{eq:kbar}, and are given below. 

\begin{equation}\label{eq:sigma_regions}
    \begin{aligned}
    \kb = 1, \text{ if } \quad &p\sigma^h \geq 1 \\
    \kb = 2, \text{ if } \quad &p\sigma^h < 1  \text{ and } p\sigma^h + (1-p)\sigma^\ell \geq 1 \\
    \kb = 3, \text{ if } \quad &p\sigma^h + (1-p)\sigma^\ell < 1 \\
    \end{aligned}
\end{equation}
Let us denote these regions as $E_1, E_2, E_3$, whose union is $\mathbb{R}_+^2$. The constructed strategies take three different forms described in Lemma \ref{lem:APA_IUI}, contingent on the value of $\kb$. Thus, there are three cases the system of equations \eqref{eq:SOE} can take. They are given below, where we seek to find multipliers $(\sigma^h,\sigma^\ell,\lambda_\mcb)$ that satisfies one of the three cases for a given instance $(A^h,A^\ell,p,B) \in \mcal{G}$ -- note that we can uniquely recover $(\lambda^h,\lambda^\ell,\lambda_\mcb)$ from a tuple $(\sigma^h,\sigma^\ell,\lambda_\mcb)$.

\begin{equation}\label{eq:SB_SOE}
    \begin{aligned}
        &\text{\underline{Case 1}} \ & \  &\text{\underline{Case 2}} \ & \  &\text{\underline{Case 3}}  \\
        \text{(i)} \ \  &\frac{1}{2p(\sigma^h)^2} = \lambda_\mcb A^h & \ &\frac{p}{2} + \frac{1 - p\sigma^h}{\sigma^\ell} = \lambda_\mcb A^h & \ &\frac{p}{2} + 1-p = \lambda_\mcb A^h  \\
        \text{(ii)} \ \ &0 = A^\ell & \  &\frac{\left(1 - p\sigma^h\right)^2}{2(1-p)(\sigma^\ell)^2} = \lambda_\mcb A^\ell & \ &\frac{1-p}{2} = \lambda_\mcb A^\ell \\
        \text{(iii)} \ \ &p\sigma^h A^h = B & \ &p\sigma^h A^h + (1-p)\sigma^\ell A^\ell = B & \ &p\sigma^h A^h + (1-p)\sigma^\ell A^\ell = B \\
        &\text{such that } & \  &\text{such that }  & \ &\text{such that }  \\
        \text{(iv)} \ \ &\bs{\sigma} \in E_1 & \ &\bs{\sigma} \in E_2 & \ &\bs{\sigma} \in E_3 \\
        \text{(v)} \ \ &\sigma^h \leq \sigma^\ell & \ &\sigma^h \leq \sigma^\ell & \ &\sigma^h \leq \sigma^\ell
    \end{aligned}\tag{$\star$}
\end{equation}
A solution $(\sigma^h,\sigma^\ell,\lambda_\mcb)$ to \eqref{eq:SB_SOE} cannot satisfy two cases simultaneously, due to the $E_i$ being disjoint. Observe that the individual battlefield values $v_j$ (whose sum total is one) do not appear in these equations. In fact, one would arrive to the system \eqref{eq:SB_SOE} when considering Lotto games with any number $n \geq 1$ of battlefields whose total value is normalized to one -- the individual battlefield values do not play a role in the analysis. For simplified exposition, we will henceforth consider players' strategies as allocations to a single battlefield of value one ($F_\mca$ and $F_\mcb$ with no $j$ dependence), noting this is mathematically equivalent to any arbitrary set of $n$ battlefield valuations $\bs{v}$ that sum to one.

Also, note the multiplier $\sigma^\ell$ does not appear in the system of Case 1, but does appear in the condition (v). Here, $\sigma^\ell$ can be set to $\infty$ to satisfy (v), without affecting other variables. We detail the complete solutions\footnote{If one considers general settings with more than two budget types, the equations \eqref{eq:SB_SOE} generally form an under-determined system for the unknown variables. Hence, there may not exist unique solutions. These challenges are considered for future work.} to \eqref{eq:SB_SOE}, and prove their associated strategies (from \eqref{eq:A_marginals}) constitute equilibria in the result below.
\begin{proposition}\label{prop:SB_equil}
	The set of game instances $G = (A^h,A^\ell,p,B)$ for which a solution to \eqref{eq:SB_SOE} exists and their corresponding equilibrium strategies and payoffs are given as follows. Define $\gamma^t := \frac{A^t}{B}$ for $t \in \{h,\ell\}$.
	
	\noindent{\textbf{\underline{Case 1}}:} 
	Suppose $\gamma^h \leq 1 \text{ and } \gamma^\ell = 0$. The solution to \eqref{eq:SB_SOE} is given by $\lambda^h = \frac{1}{2B}$, $\lambda_\mcb = \frac{pA^h}{2B^2}$, and $\lambda^\ell \geq \frac{1}{2B}$. The equilibrium strategies are
	\begin{equation}\label{eq:case1_BNE}
    	F_\mca^{h} = (1-\gamma^h)\bs{\delta}_0 + \gamma^h \text{\emph{Unif}}(0,2B), \quad F_\mca^{\ell} = \bs{\delta}_0, \quad F_\mcb = \text{\emph{Unif}}(0,2B)
    \end{equation}
    and the (ex-ante) equilibrium payoffs are
    \begin{equation}
        \pi_A = \frac{p \gamma^h}{2}, \quad\quad \pi_B = p (1-\frac{\gamma^h}{2}) + (1-p)
    \end{equation}
    
    \noindent{\textbf{\underline{Case 2}}:} Suppose $\gamma^\ell \leq H\left(\gamma^h \right)$, where $H$ is defined in \eqref{eq:G}. The unique solution to \eqref{eq:SB_SOE} is given by $\sigma^\ell = (1-\frac{B}{A^h})\sqrt{\frac{A^h/((1-p)A^\ell)}{p + (1-p)A^\ell/A^h}}$, $\sigma^h = \frac{B - (1-p)\sigma^\ell A^\ell}{p A^h}$, and $\lambda_\mcb = \frac{\left(\sqrt{(1-p)A^\ell} + \sqrt{p A^h +  (1-p)A^\ell}\right)^2}{2(A^h)^2}$. The equilibrium strategies are
    \begin{equation}\label{eq:case2_BNE}
    	\begin{aligned}
        	F_\mca^{h} &= \text{\emph{Unif}}\left(L^\ell, L^\ell + L^h \right), \quad F_\mca^{\ell} = \left(1-\frac{1-p\sigma^h}{(1-p)\sigma^\ell}\right)\bs{\delta}_0 + \frac{1-p\sigma^h}{(1-p)\sigma^\ell}\text{\emph{Unif}}\left(0,L^\ell \right)   \\
        	F_\mcb &= (1-p\sigma^h)\text{\emph{Unif}}\left(0,L^\ell \right) + p\sigma^h \text{\emph{Unif}}\left(L^\ell, L^\ell + L^h \right) \\
    	\end{aligned}
    \end{equation}
    where $L^h = \frac{p}{\lambda_\mcb}$ and $L^\ell = \frac{1-p\sigma^h}{\lambda^\ell}$. The equilibrium payoffs are
    \begin{equation}
    	\pi_A = p (1-p\sigma^h)\left(1- \frac{\sigma^h}{\sigma^\ell}\right) + \lambda_\mcb B, \quad\quad \pi_B = \lambda_\mcb B - \frac{1 - p\sigma^h}{\sigma^\ell} + (1-p)
    \end{equation}
    
    \noindent{\textbf{\underline{Case 3}}:} Suppose $\frac{A^\ell}{A^h} = \frac{1-p}{2-p} \text{ and } 2-p < \frac{A^h}{B} < 2 + \frac{p}{1-p}$. A solution to \eqref{eq:SB_SOE} is of the form $\lambda_\mcb = \frac{2-p}{2A^h}$, $\sigma^h \in \left( \frac{ \frac{B}{A^h}\left(2 + \frac{p}{1-p}- \frac{A^h}{B}\right)}{p\left(1+\frac{p}{1-p}\right)} ,\frac{\frac{B}{A^h} \left(2 + \frac{p}{1-p}\right)}{p\left(2+\frac{p}{1-p}+\frac{1-p}{p}\right)} \right)$, and $\sigma^\ell = \frac{B-p\sigma^h A^h}{(1-p)A^\ell}$. The equilibrium strategies are
    \begin{equation}\label{eq:case3_BNE}
    	\begin{aligned}
        	F_\mca^{h} &=  \text{\emph{Unif}}\left(L^\ell,L^\ell+L^h\right), \quad F_\mca^{\ell} = \text{\emph{Unif}}\left(0,L^\ell\right) \\
        	F_\mcb(x) &= (1-p\sigma^h + (1-p)\sigma^\ell)\bs{\delta}_0 + (1-p)\sigma^\ell\text{\emph{Unif}}\left(0,L^\ell\right) + p\sigma^h\text{\emph{Unif}}\left(L^\ell,L^\ell+L^h\right) \\
    	\end{aligned}
    \end{equation}
    where $L^h = \frac{p}{\lambda_\mcb}$ and $L^\ell = \frac{1-p}{\lambda_\mcb} = 2A^\ell$. The equilibrium payoffs are given by 
    \begin{equation}
    \pi_A = 1 - \lambda_\mcb B, \quad\quad \pi_B = \lambda_\mcb B.
    \end{equation}
\end{proposition}
\begin{proof}
    We divide this proof into two parts. In the first part, we detail the steps used in each Case to calculate the algebraic solution to \eqref{eq:SB_SOE} and the set of game instances for which it is valid. In the second part, we provide a proof that the corresponding strategies recovered from \eqref{eq:A_marginals} do in fact constitute an equilibrium to the BL game. We will rely on shorthand notations $\gamma_i = A_i/B$ when convenient.
    
    \noindent{\textbf{\underline{Case 1}}:}
    The solution to \eqref{eq:SB_SOE} can directly be found to be $\lambda^h = \frac{1}{2B}$, $\lambda_\mcb = \frac{pA^h}{2B^2}$, and any $\lambda^\ell \geq \frac{1}{2B}$ (to satisfy (v)). Such a solution must also satisfy (iv), $p \sigma^h = 1/A^h \geq 1$. Combined with (ii), the set of valid game parameters is $\gamma^h \leq 1$ and $\gamma^\ell = 0$: player $\mca$'s budget in type $h$ is smaller than player $\mcb$'s budget, and has a budget of zero in type $\ell$. Since $\lambda^\ell$ does not appear in the algebraic equations of \eqref{eq:SB_SOE} (only in the constraints), this is essentially unique. Plugging these values into \eqref{eq:A_marginals}, we obtain the resulting strategies. 
    
    \noindent{\textbf{\underline{Case 2}}:}
	To solve for $\lambda_\mcb$, we have $1 - p\sigma^h = \sqrt{2(1-p)\lambda_{2}\sigma_{2}A^\ell}$ from (ii). Substituting into (i), we obtain a quadratic equation in $\sqrt{\lambda_\mcb} > 0$. Its (positive) solution yields the expression for $\lambda_\mcb$.
	
	Multiplying (ii) by $\frac{A^h}{A^\ell}$, the RHS of equations (i) and (ii) become equivalent. From (iv) of \eqref{eq:SB_SOE}, we use the substitution $1-p\sigma^h = 1-
	(\gamma^h)^{-1}+ (1-p)\sigma^\ell\frac{\gamma^\ell}{\gamma^h}$ to obtain
	$\sigma^\ell = \lvert 1-(\gamma^h)^{-1}\mid\sqrt{\frac{\gamma^h/((1-p)\gamma^\ell)}{p + (1-p)\gamma^\ell/\gamma^h}}$. The condition (iv) requires $p\sigma^h < 1$. Using the substitution $\sigma^h = \frac{1-(1-p)\sigma^\ell\gamma^\ell}{p\gamma^h}$ from (iii), we deduce that $\gamma^h > 1$:
	\begin{equation}
    	\begin{aligned}
    	p\sigma^h &= (\gamma^h)^{-1}\left(1-(1-p)\gamma^\ell\lvert 1-(\gamma^h)^{-1}\rvert\cdot\sqrt{\frac{\gamma^h/((1-p)\gamma^\ell)}{p + (1-p)\gamma^\ell/\gamma^h}} \right) \\ 
    	&= (\gamma^h)^{-1} - \lvert 1-(\gamma^h)^{-1}\rvert \cdot \sqrt{\frac{(1-p)\gamma^\ell/\gamma^h}{p + (1-p)\gamma^\ell/\gamma^h}} < 1 \\
    	&\Rightarrow 1-(\gamma^h)^{-1} > -\lvert 1-(\gamma^h)^{-1}\rvert \Rightarrow \gamma^h > 1
    	\end{aligned}
	\end{equation}
	We can also deduce from (iii) and $\gamma^h \geq \gamma^\ell$ that $\gamma^\ell \leq 1$. The condition (iv) also requires $p \sigma^h + (1-p)\sigma^\ell \geq 1$. From this, we obtain 
	\begin{equation}\label{eq:F2}
	    \gamma^\ell \leq \frac{1-p}{2-p}\gamma^h.
	\end{equation}
	Furthermore, the positivity of $\sigma^\ell$ is trivially satisfied. However,  positivity of $\sigma^h$ requires that $1 - (1-p)\sigma^\ell\gamma^\ell > 0$. Plugging in the expression for $\sigma^\ell$, we obtain $\gamma^\ell (1-p)\left( 2-\gamma^h\right) > -p$.	Hence, the positivity constraint $\sigma^h > 0$ is equivalent to
	\begin{equation}\label{eq:F3}
    	\begin{aligned}
        	\gamma^h \leq 2,  \text{ or }  \gamma^h > 2 \text{ and } \gamma^\ell < \frac{p}{(1-p)\left( \gamma^h-2\right)}  .
    	\end{aligned}
	\end{equation}
	
	Lastly, the constraint (v) requires $\sigma^h \leq \sigma^\ell$. Plugging in the expression for $\sigma^\ell$, we deduce that $\gamma^\ell \left( 1 - (\gamma^h-1)^2 \right) \leq (\gamma^h-1)^2\gamma^h\frac{p}{1-p}$. The term in parentheses on the LHS is positive when $\gamma^h < 2$, and negative otherwise. Hence, we obtain
	\begin{equation}\label{eq:F1}
    	\gamma^\ell 
    	\begin{cases}
        	\leq \frac{\frac{p}{1-p}(\gamma^h-1)^2}{2-\gamma^h}, &\text{if } 1 < \gamma^h \leq 2 \\
        	\geq 0, &\text{if } \gamma^h > 2 
    	\end{cases}
	\end{equation}
	The intersection of conditions \eqref{eq:F1},\eqref{eq:F2}, and \eqref{eq:F3} on the budget parameters $A^h$ and $A^\ell$, derived directly from (iv) and (v), yields $\gamma^\ell \leq H\left(\gamma^h\right)$, where $H$ was defined in \eqref{eq:G}. This establishes the set of games for which the system \eqref{eq:SB_SOE} has a solution in Case 2.

\noindent{\underline{\textbf{Case 3}}:}
	We can directly obtain $\lambda_\mcb =  \frac{2-p}{2 A^h}$. Note that $\lambda_\mcb =  \frac{1-p}{2A^\ell}$ as well, from which we obtain $A^\ell = \frac{1-p}{2-p}A^h$. From (iii), we have $\sigma^\ell = \frac{B-p A^h\gamma^h}{(1-p)A^\ell}$. Substituting this in the condition (iv), $p\sigma^h + (1-p)\sigma^\ell<1$, we obtain $\sigma^h > \frac{(\gamma^h)^{-1}\left(2 + \frac{p}{1-p}-\gamma^h\right)}{p\left(1+\frac{p}{1-p}\right)}$. Similarly, constraint (v), $\sigma^h \leq \sigma^\ell$, yields $\sigma^h \leq \frac{(\gamma^h)^{-1}\left(2 + \frac{p}{1-p}\right)}{p\left(2+\frac{p}{1-p}+\frac{1-p}{p}\right)}$. A feasible $\sigma^h$ exists within these constraints if and only if $\gamma^h > 1 + \frac{1 + \frac{1-p}{p}}{2 + \frac{p}{1-p} + \frac{1-p}{p}} = 2-p$ (upper bound must be larger than lower bound), and $\gamma^h < 2 + \frac{p}{1-p}$ (lower bound must be positive). Subsequently, \eqref{eq:A_marginals} recovers the strategies \eqref{eq:case3_BNE}. The union of characterized parameter sets in all three cases constitutes the $\mcal{R}_5$ region in Theorem \ref{thm:BNE_regions}.
	
	\noindent\underline{\textbf{Part 2}}: We prove the strategy profile $(F_\mca,F_\mcb)$ recovered from \eqref{eq:A_marginals} is an equilibrium\footnote{This proof can be extended to scenarios where player $\mca$ has an arbitrary number $m$ endowment types. That is, if one can derive a solution to the system of $m+1$ equations (instead of just 3 in \eqref{eq:SOE}), the profile $(F_\mca,F_\mcb)$ one constructs using Siegel's algorithm (in analogous manner to Section \ref{sec:APA_IUI}) is an equilibrium to the BL game. Characterizing solutions to a system of $m+1$ equations, however, is a non-trivial extension that we leave for future study.}. We can immediately deduce the strategies in Case 1 are equilibria to the BL game by observing that player $\mca$ has zero budget in type $\ell$, and $(F_\mca^{h},F_\mcb)$ forms the unique equilibrium to the complete information General Lotto game \cite{Hart_2008} with a single battlefield of value $p$. We will focus on the strategies produced from Case 2, as the proof for Case 3  follows analogous arguments. 
	
	We first calculate the (ex-interim) payoffs from the strategies \eqref{eq:case2_BNE}.
	\begin{equation}
	    \begin{aligned}
	        U_A(F_\mca^{h},F_\mcb) &= \int_{0}^{\infty} F_\mcb(x) \,dF_\mca^{h} = \int_{L^\ell}^{L^\ell+L^h} \left[ L^\ell \lambda^\ell + \lambda^h(x - L^\ell) \right] \frac{\lambda_\mcb}{p} dx \\
	        &= (1-p \sigma^h)\left(1- \frac{\sigma^h}{\sigma^\ell} \right) + \lambda^h A^h \\
	        U_A(F_\mca^{\ell},F_\mcb) &= \int_{0}^{\infty} F_\mcb(x) \,dF_\mca^{\ell} = \int_{0}^{L^\ell} \lambda^\ell x \cdot \frac{\lambda_\mcb}{1-p} dx = \lambda^\ell A^\ell
	    \end{aligned}
	\end{equation}
	The expected payoff (first equation of \eqref{eq:Bayesian_payoff}) to player $\mca$ is then $\pi_A = p (1-p\sigma^h)\left(1- \frac{\sigma^h}{\sigma^\ell}\right) + \lambda_\mcb B$ (using (iii)). The payoff to player $\mcb$ is $\pi_B = 1 - \pi_A$. We need to show $F_\mca$ is a best-response to $F_\mcb$, and vice versa.
	
	For any $F_\mca'^{h} \in \mcal{L}(A^h)$, the payoff in type $h$ is
	\begin{equation}
	    \begin{aligned}
	        U_A(F_\mca'^{h},F_\mcb) &= \int_0^{L^\ell} \lambda^\ell x \,dF_\mca'^{h} + \int_{L^\ell}^{L^\ell+L^h}\left[ L^\ell \lambda^\ell + \lambda^h(x - L^\ell) \right]\,dF_\mca'^{h} + \int_{L^\ell + L^h}^\infty dF_\mca'^{h} \\
	        &= (\lambda^\ell - \lambda^h)\left( \int_0^{L^\ell} x\,dF_\mca'^{h} - L^\ell \int_0^{L^\ell} dF_\mca'^{h} \right) \\ 
	        &\quad + \lambda^h \left( (L^\ell + L^h)\int_{L^\ell}^{L^\ell+L^h} dF_\mca'^{h} -  \int_{L^\ell}^{L^\ell+L^h} x\,dF_\mca'^{h} \right) + \lambda^h A^h + L^\ell(\lambda^\ell-\lambda^h) \\
	        &\leq \lambda^h A^h + L^\ell(\lambda^\ell-\lambda^h) = \lambda^h A^h + (1-\frac{\lambda^h}{\lambda^\ell})(1-p\sigma^h).
	    \end{aligned}
	\end{equation}
	In the second equality, we used the identities $\int_{L^\ell}^{L^\ell+L^h} x\,dF_\mca'^{h} = A^h - \int_0^{L^\ell} x\,dF_\mca'^{h} - \int_{L^\ell+L^h}^\infty x\,dF_\mca'^{h}$ and $\int_{L^\ell}^{L^\ell+L^h} dF_\mca'^{h} = 1 - \int_0^{L^\ell} dF_\mca'^{h} - \int_{L^\ell+L^h}^\infty dF_\mca'^{h}$. The inequality follows from two applications of Markov's inequality: $\int_0^{L^\ell} x\,dF_\mca'^{h} \leq L^\ell\int_0^{L^\ell} dF_\mca'^{h}$, and $-  \int_{L^\ell}^{L^\ell+L^h} x\,dF_\mca'^{h} \leq -(L^\ell+L^h) \int_{L^\ell}^{L^\ell+L^h} dF_\mca'^{h}$. Hence, the payoff in type $h$ is upper-bounded by $\lambda^h A^h + (1-\frac{\lambda^h}{\lambda^\ell})(1-p\sigma^h)$, which can be attained whenever $\text{supp}(F_\mca'^{h}) \subseteq [L^\ell,L^\ell+L^h]$ (for which the Markov inequalities hold with equality). We provide analogous calculations for any $F_\mca'^{\ell} \in \mcal{L}(A^\ell)$:
	\begin{equation}
	    \begin{aligned}
	        U_A(F_\mca'^{\ell},F_\mcb) &= \int_0^{L^\ell} \lambda^\ell x \,dF_\mca'^{\ell} + \int_{L^\ell}^{L^\ell+L^h}\left[ L^\ell \lambda^\ell + \lambda^h(x - L^\ell) \right]\,dF_\mca'^{\ell} + \int_{L^\ell + L^h}^\infty dF_\mca'^{\ell} \\
	        &= (\lambda^\ell - \lambda^h)\left(L^\ell \int_{L^\ell}^{L^\ell+L^h} dF_\mca'^{\ell} - \int_{L^\ell}^{L^\ell+L^h} x\,dF_\mca'^{\ell} \right) \\
	        &\quad + \left( \int_{L^\ell+L^h}^\infty dF_\mca'^{\ell} - \lambda^\ell \int_{L^\ell+L^h}^\infty dF_\mca'^{\ell} \right) + \lambda^\ell A^\ell \\
	        &\leq -p(\sigma^\ell - \sigma^h) \int_{L^\ell+L^h}^\infty dF_\mca'^{\ell} + \lambda^\ell A^\ell \\
	        &\leq \lambda^\ell A^\ell
	    \end{aligned}
	\end{equation}
	The first inequality is similarly obtained from two applications of Markov's inequality. The second inequality follows from the condition (v), $\sigma^\ell \geq \sigma^h$. Hence, the payoff in type $\ell$ is upper-bounded by $\lambda^\ell A^\ell$, which can be attained whenever $\text{supp}(F_\mca'^{\ell}) \subseteq [0,L^\ell]$. The strategy $F_\mca$ satisfies these properties, and hence is a best-response to $F_\mcb$.

	For any $F_\mcb' \in \mcal{L}(B)$, player $\mcb$'s expected payoff \eqref{eq:Bayesian_payoff} is
	\begin{equation}
	    \begin{aligned}
	        \Pi_B(F_\mcb',F_\mca) &= p\left[ \int_{L^\ell}^{L^\ell+L^h} \frac{\lambda_\mcb}{p}(x-L^\ell) \,dF_\mcb' + \int_{L^\ell+L^h}^\infty dF_\mcb' \right] \\
	        &\quad + (1-p)\left[ \int_{0}^{L^\ell} \left(1 - \frac{\lambda_\mcb L^\ell}{1-p} + \frac{\lambda_\mcb}{1-p}x \right) \,dF_\mcb' + \int_{L^\ell}^\infty dF_\mcb' \right] \\
	        &= \lambda_\mcb \left( \int_0^{L^\ell+L^h} x\,dF_\mcb' - L^\ell \int_0^{L^\ell+L^h} dF_\mcb' \right) + p\int_{L^\ell+L^h}^\infty dF_\mcb' + (1-p) \\
	        &= \lambda_\mcb B - \lambda_\mcb L^\ell + (1-p) + \lambda_\mcb\left( (L^h + L^\ell) \int_{L^\ell+L^h}^\infty dF_\mcb' - \int_{L^\ell+L^h}^\infty x \,dF_\mcb' \right) \\
	        &\leq \lambda_\mcb B - \lambda_\mcb L^\ell + (1-p) = \pi_B
	    \end{aligned}
	\end{equation}
	Player B's expected payoff is upper-bounded by $\pi_B$, which can be attained for any strategy with $\text{supp}(F_\mcb') \subseteq [0,L^\ell+L^h]$. Because $F_\mcb$ is one such strategy, it is a best-response to $F_\mca$.
\end{proof}

\subsection{Equilibria in regions $\mcal{R}_1$ - $\mcal{R}_4$}\label{sec:other_regions}

Here, we provide all derivations of equilibria corresponding to the payoffs in regions $\mcal{R}_i$, $i=1,\ldots,4$ (Theorem \ref{thm:BNE_regions}). As shown in Proposition \ref{prop:SB_equil}, such equilibria cannot be generated through the methods of Section \ref{sec:connection_method}. 

To first give some informal intuition, we provide some descriptions of the equilibria in these regions. The equilibrium strategies in $\mcal{R}_1$ and $\mcal{R}_2$ are convex combinations between an equilibrium strategy on the border\footnote{Since equilibria on the border are not necessarily unique, i.e. Case 3 parameters of Proposition \ref{prop:SB_equil}, the equilibria in the regions $\mcal{R}_1$ and $\mcal{R}_2$ are not unique. However, all equilibria in one game instance yield identical payoffs, since it is a constant sum game (in ex-ante payoffs).}  of $\mcal{R}_5$ and equilibrium in its corresponding benchmark complete information game $\text{GL}(\bar{A},B,0)$. As a result, the equilibrium payoff for any $G$ in $\mcal{R}_1$ or $\mcal{R}_2$ coincides with the equilibrium payoff of its corresponding benchmark game. In regions $\mcal{R}_3$ and $\mcal{R}_4$, the ``high" budget $A^h$ is disproportionately higher than the ``low" budget $A^\ell$. We find an equilibrium strategy for the uninformed player is to not compete against the high budget at all, thus giving up a payoff $p$ to the obfuscating player. In the forthcoming proofs, we make extensive use of Markov's inequality: $\int_a^b x \,dF \leq b \int_a^b dF$ for any distribution $F$.

\noindent\underline{\textbf{Region } $\mcal{R}_3$}: Suppose the game parameters belong to region 
\begin{equation}
	\mcal{R}_3 = \left\{ (\gamma^h, \gamma^\ell)\in\mcal{R} : \gamma^h \geq 2 + \frac{p}{1-p} \text{ and } 1 \leq \gamma^\ell \leq \frac{1-p}{2-p}\gamma^h \right\}.
\end{equation}
The following is an equilibrium.
\begin{equation}\label{eq:region3_BNE}
    F_\mca^{h} =  \text{Unif}\left(2A^\ell,2(A^h-A^\ell) \right), \quad F_\mca^{\ell} = \text{Unif}\left(0,2A^\ell\right), \quad F_\mcb = (1-(\gamma^\ell)^{-1})\bs{\delta}_0 + (\gamma^\ell)^{-1}\text{Unif}\left(0,2A^\ell\right)
\end{equation}
The equilibrium payoff is given by $\pi_A = p + (1-p)\left(1 - \frac{1}{2\gamma^\ell}\right)$.
\begin{proof}
	First, we show $F_\mca$ is a best-response to $F_\mcb$. For any $\left\{F_\mca'^{(i)} \in \mathcal{L}(A_i)\right\}_{t=h,\ell}$, player $\mca$'s expected payoff is
	\begin{equation}
    	\begin{aligned}
    	&p \left[ \int_{0}^{2A^\ell} \left( 1 - (\gamma^\ell)^{-1} + \frac{(\gamma^\ell)^{-1}}{2A^\ell}x \right) \,dF_\mca'^{h} + \int_{2A^\ell}^\infty dF_\mca'^{h} \right] \\
    	&\quad + (1-p)\left[ \int_{0}^{2A^\ell} \left( 1 - (\gamma^\ell)^{-1} + \frac{(\gamma^\ell)^{-1}}{2A^\ell}x \right) \,dF_\mca'^{\ell} + \int_{2A^\ell}^\infty dF_\mca'^{\ell} \right] \\
    	&\leq p + (1-p) \left[ \int_{0}^{2A^\ell} \left( 1 - (\gamma^\ell)^{-1} + \frac{(\gamma^\ell)^{-1}}{2A^\ell}x \right) \,dF_\mca'^{\ell} + \int_{2A^\ell}^\infty dF_\mca'^{\ell} \right] \\
    	\end{aligned}
	\end{equation}
	The inequality follows by selecting any $F_\mca'^{h}$ such that $\text{supp}(F_\mca'^{h}) \subset [2A^\ell,\infty)$, which awards player $\mca$ the payoff $p$ from state 1 outright. This is possible because $\gamma^h \geq 2 + \frac{p}{1-p} > 2$, from the assumption. The expression above can be re-written and upper-bounded as follows:
	\begin{equation}
    	\begin{aligned}
    	&p + (1-p)\left[\left( 1 - (\gamma^\ell)^{-1}\right) \int_0^{2A^\ell}dF_\mca'^{\ell} + \frac{(\gamma^\ell)^{-1}}{2} + \int_{2A^\ell}^\infty dF_\mca'^{\ell} - \frac{(\gamma^\ell)^{-1}}{2A^\ell} \int_{2A^\ell}^\infty x\,dF_\mca'^{\ell} \right] \\
    	&\quad\leq p + (1-p) \left( 1 - \frac{(\gamma^\ell)^{-1}}{2}\right)
    	\end{aligned}
	\end{equation}
	The inequality holds with equality if and only if $\text{supp}(F_\mca'^{\ell}) \subseteq [0,2A^\ell]$. We have thus established an upper bound on A's payoff to $F_\mcb$ that is achieved by $F_\mca$.
	
	Now we show $F_\mcb$ is a best-response to $F_\mca$. Let  $K := (1-p) - \frac{p \gamma^\ell}{\gamma^h-2\gamma^\ell} \geq 0$, which is non-negative due to the assumption $\gamma^\ell \leq \frac{1-p}{2-p}\gamma^h$. For any $F_\mcb' \in \mathcal{L}(B)$, player $\mcb$'s payoff is
	\begin{equation}
	\begin{aligned}
	&p \left[ \int_{2A^\ell}^{2(A^h-A^\ell)} \frac{x-2A^\ell}{2(A^h-2A^\ell)} \,dF_\mcb' + \int_{2(A^h-A^\ell)}^\infty dF_\mcb' \right] + (1-p) \left[ \int_{0}^{2A^\ell} \frac{x}{2A^\ell} \,dF_\mcb' + \int_{2A^\ell}^\infty dF_\mcb' \right] \\
	&= \frac{1-p}{2A^\ell}\int_0^{2A^\ell} x \,dF_\mcb' + \frac{p}{2(A^h-2A^\ell)}\int_{2A^\ell}^{2(A^h-A^\ell)} x\,dF_\mcb' - \left( \frac{p A^\ell}{A^h - 2A^\ell} - (1-p) \right) \int_{2A^\ell}^{2(A^h-A^\ell)} \,dF_\mcb' \\
	&\quad\quad + \int_{2(A^h-A^\ell)}^\infty dF_\mcb' \\
	\end{aligned}
	\end{equation}
	Applying the identity $\int_{0}^{2A^\ell} x\,dF_\mcb' = B - \int_{2A^\ell}^{2(A^h-A^\ell)} x\,dF_\mcb' - \int_{2(A^h-A^\ell)}^\infty x\,dF_\mcb'$, we then obtain
	\begin{equation}
	\begin{aligned}
	&= K\left( - \frac{1}{2A^\ell} \int_{2A^\ell}^{2(A^h-A^\ell)} x\,dF_\mcb' + \int_{2A^\ell}^{2(A^h-A^\ell)} \,dF_\mcb' \right) + (1-p)\frac{(\gamma^\ell)^{-1}}{2} \\
	&\quad + \int_{2(A^h-A^\ell)}^\infty dF_\mcb' - \frac{1-p}{2A^\ell}\int_{2(A^h-A^\ell)}^\infty x \,dF_\mcb' \\
	&\leq  (1-p)\frac{(\gamma^\ell)^{-1}}{2} + \left( p + (1-p)\left(2 - \frac{\gamma^h}{\gamma^\ell} \right) \right) \int_{2(A^h-A^\ell)}^\infty dF_\mcb' \\
	&\leq  (1-p)\frac{(\gamma^\ell)^{-1}}{2}
	\end{aligned}
	\end{equation}
	The first inequality results from applying Markov's inequality to $\int_{2A^\ell}^{2(A^h-A^\ell)} x\,dF_\mcb'$ and $\int_{2(A^h-A^\ell)}^\infty x\,dF_\mcb'$.  The second inequality follows from non-positivity of term in parentheses (from assumption of the Lemma). This inequality holds with equality if and only if $\text{supp}(F_\mcb') \subseteq [0,2A^\ell]$. We have thus established an upper bound on player $\mcb$'s payoff to $F_\mca$ that is achieved by $F_\mcb$.

\end{proof}

\noindent\underline{\textbf{Region } $\mcal{R}_4$}: Suppose the game parameters belong to region 
\begin{equation}
	\mcal{R}_4 := \left\{ (\gamma^h, \gamma^\ell)\in\mcal{R} : \gamma^h \geq 2 + \frac{p}{1-p} \text{ and } \frac{p}{(1-p)(\gamma^h-2)} \leq \gamma^\ell \leq 1 \right\}.
\end{equation}
The following is an equilibrium: $F_\mca^{h} = \text{Unif}\left(2B,2(A^h-B)\right)$,  $F_\mca^{\ell} = (1-\gamma^\ell)\bs{\delta}_0 + \gamma^\ell\text{Unif}\left(0,2B\right)$, $F_\mcb = \text{Unif}\left(0,2B\right)$, and the equilibrium payoff is given by $\pi_A(G) = p + (1-p)\frac{\gamma^\ell}{2}$.

\begin{proof}
	First, show $F_\mca$ is a best-response to $F_\mcb$. player $\mca$'s payoff for any $\left\{F_\mca'^{(i)} \in \mathcal{L}(A_i)\right\}_{t=h,\ell}$ is
	\begin{equation}
	\begin{aligned}
	&p \left[ \int_{0}^{2B}\frac{x}{2B} \,dF_\mca'^{h} + \int_{2B}^\infty dF_\mca'^{h} \right] + (1-p) \left[ \int_{0}^{2A^\ell}\frac{x}{2B} \,dF_\mca'^{\ell} + \int_{2B}^\infty dF_\mca'^{\ell} \right] \\
	&\leq p + (1-p) \left[ \int_{0}^{2A^\ell}\frac{x}{2B} \,dF_\mca'^{\ell} + \int_{2B}^\infty dF_\mca'^{\ell} \right] \\
	&= p + (1-p) \left[ \frac{\gamma^\ell}{2} + \int_{2B}^\infty dF_\mca'^{\ell} - \frac{1}{2B}\int_{2B}^\infty x\,dF_\mca'^{\ell} \right] \\
	&\leq p + (1-p)\frac{\gamma^\ell}{2}.
	\end{aligned}
	\end{equation}
	The first inequality follows by selecting any $F_\mca'^{h}$ such that $\text{supp}(F_\mca'^{h}) \subset [2B,\infty)$, which awards player $\mca$ the payoff $p$ outright. This is possible because $\gamma^h \geq 2$, from the assumption. The second inequality follows by applying Markov's inequality to $\int_{2B}^\infty x\,dF_\mca'^{\ell}$. This inequality holds if and only if $\text{supp}(F_\mca'^{\ell}) \subseteq [0,2B]$.
	
	Now we show $F_\mcb$ is a best-response to $F_\mca$. For any $F_\mcb \in \mathcal{L}(B)$, player $\mcb$'s payoff is
	\begin{equation}
    	\begin{aligned}
    	&\frac{(1-p)\gamma^\ell}{2B}\int_0^{2B} x\,dF_\mcb' + (1-p)\left(1 - \gamma^\ell\right)\int_0^{2B} dF_\mcb' \\
    	&\quad + \frac{p}{2(A^h-2B)}\int_{2B}^{2(A^h-B)} x\,dF_\mcb' + \left((1-p)\gamma^\ell - \frac{p}{\gamma^h-2}\right) \int_{2B}^{2(A^h-B)} dF_\mcb' + \int_{2(A^h-B)}^\infty dF_\mcb' \\
    	&= K\left( \int_{2B}^{2(A^h-B)} dF_\mcb' - \frac{1}{2B}\int_{2B}^{2(A^h-B)}  x\,dF_\mcb' \right) \\
    	&\quad + \left(p + (1-p)\gamma^\ell\right)\int_{2(A^h-B)}^\infty dF_\mcb'  - \frac{(1-p) \gamma^\ell}{2B}\int_{2(A^h-B)}^\infty x\,dF_\mcb' + (1-p)\left(1 - \frac{\gamma^\ell}{2}\right) \\
    	&\leq (1-p)\left(1 - \frac{\gamma^\ell}{2}\right)
    	\end{aligned}
	\end{equation}
	where $K := (1-p)\gamma^\ell - \frac{p}{\gamma^h-2} > 0$, from the assumption. In the equality, we substituted $\int_0^{2B} x\,dF_\mcb' = B - \int_{2B}^{2(A^h-B)} x\,dF_\mcb' - \int_{2(A^h-B)}^\infty x\,dF_\mcb'$. The inequality is due to applying Markov's inequality and noting the expression $p + (1-p)\gamma^\ell - (1-p)\gamma^\ell\left(\gamma^h - 1  \right) = p + (1-p)\gamma^\ell \left(2 - \gamma^h\right) \leq 0$ holds from the assumptions. This holds with equality if and only if $\text{supp}(F_\mcb') \subseteq [0,2B]$.
\end{proof}

\noindent\underline{\textbf{Regions } $\mcal{R}_1$ and $\mcal{R}_2$}:

Consider the set of $(\gamma^h,\gamma^\ell) \in \mcal{R}_1 \cup \mcal{R}_2$ that have a fixed average budget $\bar{\gamma}$. Define 
\begin{equation}
\bs{\gamma}^{\text{bd}} := 
\begin{cases}
(\bar{\gamma}/p,0) \in \mcal{R}, &\text{if } \bar{\gamma} \leq p \\
(2 - p/\bar{\gamma},H(2-p/\bar{\gamma})) \in \mcal{R}, &\text{if } p < \bar{\gamma} \leq 1 \\
((2-p)\bar{\gamma}, (1-p)\bar{\gamma}) \in \mcal{R}, &\text{if } 1 < \bar{\gamma}
\end{cases}
\end{equation}
where $H$ is defined in \eqref{eq:G} and $\mcal{R} = \{(\gamma^h,\gamma^\ell) : \gamma^h \geq \gamma^\ell\}$. The points $\bs{\gamma}^{\text{bd}}$ specified above for $\bar{\gamma} \leq 1$ are on the border of $\mcal{R}_5$, whose equilibria are given in Proposition \ref{prop:SB_equil}. The points for $1 < \bar{\gamma}$ are on the upper border of $\mcal{R}_3$, where a equilibrium is given in \eqref{eq:region3_BNE}. Define
\begin{equation}
F_\mca^{\text{bd}} := 
    \begin{cases}
    \text{given by } \eqref{eq:case1_BNE} \text{ at } \bs{\gamma}^{\text{bd}}, &\text{if } \bar{\gamma} \leq p \\
    \text{given by } \eqref{eq:case2_BNE} \text{ at } \bs{\gamma}^{\text{bd}}, &\text{if } p < \bar{\gamma} \leq 1 \\
    \text{given by } \eqref{eq:case3_BNE} \text{ at } \bs{\gamma}^{\text{bd}}, &\text{if } 1 < \bar{\gamma}
    \end{cases}
\end{equation}
Here, $F_\mca^{\text{bd}}$ is an equilibrium strategy for player $\mca$ at the boundary point $\bs{\gamma}^{\text{bd}}$. Let us also define $(\bar{F}_A,\bar{F}_B)$ as the equilibrium at $(\bar{\gamma},\bar{\gamma}) \in \mcal{R}$, which is simply the equilibrium in the corresponding complete information game. That is,
\begin{equation}
    \bar{F}_A = 
    \begin{cases}
    (1-\bar{\gamma})\delta_0 + \bar{\gamma}\text{Unif}([0,2B]), &\text{if } \bar{\gamma} \leq 1 \\
    \text{Unif}([0,2\bar{A}]), &\text{if } \bar{\gamma} > 1 \\
    \end{cases}, \quad 
    \bar{F}_B = 
    \begin{cases}
    \text{Unif}([0,2B]), &\text{if } \bar{\gamma} \leq 1 \\
    (1-\bar{\gamma}^{-1})\delta_0  + \bar{\gamma}^{-1}\text{Unif}([0,2\bar{A}]), &\text{if } \bar{\gamma} > 1 \\
    \end{cases}
\end{equation}

Let $\alpha \in [0,1]$ be the appropriate scaling that gives $\alpha \bs{\gamma}^{\text{bd}} + (1-\alpha)\cdot (\bar{\gamma},\bar{\gamma}) = (\gamma^h,\gamma^\ell)$. We claim the strategy profile $(\alpha F_\mca^{\text{bd}} + (1-\alpha) \bar{F}_A, \bar{F}_B)$ is an equilibrium for $(\gamma^h,\gamma^\ell) \in \mcal{R}_1 \cup \mcal{R}_2$.

\begin{proof}
	Since we know that $(\bar{F}_A,\bar{F}_B)$ is an equilibrium, it will suffice to show that $(F_\mca^{\text{bd}},\bar{F}_B)$ is also an equilibrium for all $\bar{\gamma}$. After lengthy calculations, we see that the payoffs from $(F_\mca^{\text{bd}},\bar{F}_B)$ coincide with the payoffs from $(\bar{F}_A,\bar{F}_B)$.

	For $\bar{\gamma} \leq p$, the equilibrium at $\bs{\gamma}^{\text{bd}}$ is given by Case 1 \eqref{eq:case1_BNE}, where player $\mcb$'s strategy is precisely $\bar{F}_B$. For $p < \bar{\gamma} \leq 1$, the equilibrium at $\bs{\gamma}^{\text{bd}}$ is given by Case 2 \eqref{eq:case2_BNE}, where it is also true that player $\mcb$'s strategy is $\bar{F}_B$ (on these border points, $\sigma^h = \sigma^\ell$). For $1 < \bar{\gamma} \leq \frac{1}{1-p}$, the equilibrium at $\bs{\gamma}^{\text{bd}}$ is given by Case 3 \eqref{eq:case3_BNE}. We note that although player $\mcb$'s equilibrium strategy here is not unique, $\bar{F}_B$ is one such strategy. 
	
	Lastly, for $\frac{1}{1-p} < \bar{\gamma}$, $(F_\mca^{\text{bd}},\bar{F}_B)$ is an equilibrium at $\bs{\gamma}^{\text{bd}}$, where $F_\mca^{\text{bd}}$ is player $\mca$'s equilibrium strategy at the border of $\mcal{R}_3$ \eqref{eq:region3_BNE}. This strategy is also identical to the monotonic equilibrium strategy from Case 3 \eqref{eq:case3_BNE}. Hence, the proof that $(F_\mca^{\text{bd}},\bar{F}_B)$ is an equilibrium follows from the analysis in Proposition \ref{prop:SB_equil}. Note that player $\mcb$'s $\mcal{R}_3$ equilibrium strategy written in \eqref{eq:region3_BNE} is not $\bar{F}_B$. Indeed, $\bar{F}_B$ in general is not an equilibrium strategy in the interior of $\mcal{R}_3$. At the border however, we know of at least two equilibria (giving the same payoffs), one of them being $(F_\mca^{\text{bd}},\bar{F}_B)$.
	
\end{proof}

\section{Optimal deterministic assignment -- per-unit cost setting}\label{sec:deterministic_proofs}


While the fixed budget setting of $\text{CAP}_\text{d}$ has previously been studied in the literature \cite{Kovenock_2020}, the per-unit cost setting has not. To derive the solution of $\text{CAP}_\text{d}$ under the per-unit cost setting, i.e. $A_C = B_1 = B_2 = \infty$ with $c,c_1,c_2 > 0$, we proceed with backwards induction. Since Theorem \ref{thm:BNE_regions} gives the equilibrium payoff of any BL game, we can readily evaluate the final payoffs in Stage 3. Thus, we must first solve for the optimal opponent investment decision in Stage 2. Recall that by the beginning of Stage 2, the commander has chosen a deterministic assignment $(A_1,A_2)$. The Stage 2 decision problem for opponent $\mcb_i$ is therefore
\begin{equation}\label{eq:Bi_CAPd_problem}
    \max_{B_i \geq 0} \left\{ \phi_i\cdot \pi^\text{CI}(B_i,A_i) - c_i B_i \right\}
\end{equation}
where $\pi^\text{CI}$ is defined in \eqref{eq:Lotto_CI}. Here, a complete information General Lotto game is played at Stage 3 since the opponents have observed the deterministic assignment $(A_1,A_2)$.

\begin{lemma}\label{lem:Bstar_CAPd}
    Consider the Stage 2 decision for opponent $\mcb_i$ in $\text{CAP}_\text{d}$ \eqref{eq:Bi_CAPd_problem}. The optimal investment for $\mcb_i$ is given by
    \begin{equation}\label{eq:Bstar_CAPd}
        B_i^*(A_i) = 
        \begin{cases}
            \sqrt{\frac{\phi_i A_i}{2 c_i}}, &\text{if } c_i < \frac{\phi_i}{2A_i} \\
            0, &\text{otherwise}
        \end{cases}
    \end{equation}
\end{lemma}
Recall we are using the assumption that $\mcb_i$ chooses the smallest investment among multiple maximizers of \eqref{eq:Bi_CAPd_problem}. In Lemma \ref{lem:Bstar_CAPd}, this only arises if $c_i = \frac{\phi_i}{2A_i}$.

Now that we have established the optimal decision in Stage 2, we can address the commander's assignment decision problem in Stage 1. As an intermediate step, we first consider the optimal assignment problem when the commander has a fixed budget $A$ and zero cost $c = 0$:

\begin{equation}\label{eq:zerocost_deterministic_commander}
    \max_{A_1 \in [0,A]} \left\{ \phi_1\cdot\pi^\text{CI}(A_1,B_1^*) + \phi_2\cdot\pi^\text{CI}(A-A_1,B_2^*)  \right\}
\end{equation}

Below, we give the optimal assignment of \eqref{eq:zerocost_deterministic_commander}.

\begin{lemma}\label{lem:commander_det_nocost}
	The optimal deterministic assignment to \eqref{eq:zerocost_deterministic_commander} is given by the following cases. Define $Q_i := \sqrt{\frac{c_i\phi_i({A}-\frac{\phi_{-i}}{2c_{-i}})}{2}}$ for $i=1,2$.
	
	\noindent\textbf{Case 1:} ${A} < \min_{i=1,2}\{\frac{\phi_i}{2c_i}\}$. Then ${A}_1^* = \frac{c_1\phi_1}{c_1\phi_1+c_2\phi_2}{A}$ and $W_\text{d}^* = \sqrt{\frac{{A}(c_1\phi_1+c_2\phi_2)}{2}}$.
	
	\noindent\textbf{Case 2:} $\min_{i=1,2}\{\frac{\phi_i}{2c_i}\} \leq {A} < \max_{i=1,2}\{\frac{\phi_i}{2c_i}\}$. Let $j = \argmin{i\in\{1,2\}} \{\frac{\phi_i}{2c_i}\} $. If ${A} \geq \frac{c_1\phi_1+c_2\phi_2}{2c_j^2}$, then ${A}_1^* = \frac{\phi_j}{2c_j}$ and $W_\text{d}^* = \phi_j+Q_{-j}$. If ${A} < \frac{c_1\phi_1+c_2\phi_2}{2c_j^2}$, then
	\begin{equation}
	\begin{aligned}
	{A}_j^* &= 
	\begin{cases}
	\frac{\phi_j}{2c_j}, &\text{if } \phi_j+Q_{-j} \geq \sqrt{\frac{{A}(c_1\phi_1+c_2\phi_2)}{2}}  \\
	\frac{c_j\phi_j}{c_1\phi_1+c_2\phi_2}{A}, &\text{if } \phi_j+Q_{-j} < \sqrt{\frac{{A}(c_1\phi_1+c_2\phi_2)}{2}} \\
	\end{cases} \\
	W_\text{d}({A}_1^*) &= \max\left\{\phi_j+Q_{-j}, \sqrt{\frac{{A}(c_1\phi_1+c_2\phi_2)}{2}} \right\}
	\end{aligned}
	\end{equation}
	
	\noindent\textbf{Case 3:} $\max_{i=1,2}\{\frac{\phi_i}{2c_i}\} \leq {A} < \frac{\phi_1}{2c_1} + \frac{\phi_2}{2c_2}$. Then 
	\begin{equation}
	\begin{aligned}
	{A}_j^* = 
	\begin{cases}
	{A} - \frac{\phi_{-j}}{2c_{-j}}, &\text{if } \phi_{-j}+Q_j \geq \phi_j+Q_{-j} \\
	\frac{\phi_j}{2c_j}, &\text{if } \phi_{-j}+Q_j < \phi_j+Q_{-j} \\
	\end{cases} \quad \text{and} \quad
	W_\text{d}^* = \max_{i=1,2} \left\{\phi_i + Q_{-i} \right\}
	\end{aligned}
	\end{equation}
	
	\noindent\textbf{Case 4:} $\frac{\phi_1}{2c_1} + \frac{\phi_2}{2c_2} \leq {A}$. Then ${A}_1^* \in [\frac{\phi_1}{2c_1},{A} - \frac{\phi_2}{2c_2})$ and $W_\text{d}^* = \phi_1+\phi_2$.
	
\end{lemma}
\begin{proof}
	Let us denote the objective function of \eqref{eq:zerocost_deterministic_commander} as $W_\text{d}({A}_1)$ for  ${A}_1 \in [0,{A}]$.
	
	\noindent\textbf{Case 1:} The objective is written as $W_\text{d}({A}_1) = \sqrt{\frac{c_1\phi_1 {A}_1}{2}} + \sqrt{\frac{c_2\phi_2({A}-{A}_1)}{2}}$    for all ${A}_1 \in [0,{A}]$. This function is concave and the derivative is zero at ${A}_1^* = \frac{c_1\phi_1}{c_1\phi_1+c_2\phi_2}{A} \in [0,{A}]$.
	
	\noindent\textbf{Case 2:} Suppose $\frac{\phi_1}{2c_1} \leq \frac{\phi_2}{2c_2}$ (analogous arguments when the opposite holds). The objective is written as
	\begin{equation}
	W_\text{d}({A}_1) = 
	\begin{cases}
	\sqrt{\frac{c_1\phi_1 {A}_1}{2}} + \sqrt{\frac{c_2\phi_2({A}-{A}_1)}{2}}, &\text{if } {A}_1 \in [0,\frac{\phi_1}{2c_1}) \\
	\phi_1 + \sqrt{\frac{c_2\phi_2({A}-{A}_1)}{2}}, &\text{if } {A}_1 \in [\frac{\phi_1}{2c_1},{A}]
	\end{cases}
	\end{equation}
	It is strictly decreasing on the second interval and takes its highest value $\phi_1 + Q_2$ at ${A}_1 = \frac{\phi_1}{2c_1}$. It is strictly increasing on the first interval if $c_1 < \frac{1}{2}\left[\frac{\phi_1}{2{A}} + \sqrt{\left(\frac{\phi_1}{2{A}} \right)^2 + \frac{2c_{2}\phi_{2}}{{A}}} \right]$ (equivalently ${A} \geq \frac{c_1\phi_1+c_2\phi_2}{2c_1^2}$), taking its highest value $\frac{\phi_1}{2} + Q_2 < \phi_1 + Q_2$. When ${A} < \frac{c_1\phi_1+c_2\phi_2}{2c_1^2}$, it takes its highest value in the first interval, $\sqrt{\frac{{A}(c_1\phi_1+c_2\phi_2)}{2}}$, at ${A}_1 = \frac{c_1\phi_1}{c_1\phi_1+c_2\phi_2} {A}$.

	
	\noindent\textbf{Case 3:} The objective is written as
	\begin{equation}
	W_\text{d}({A}_1) = 
	\begin{cases}
	\sqrt{\frac{c_1\phi_1 {A}_1}{2}} + \phi_2, &\text{if } {A}_1 \in [0,{A} - \frac{\phi_2}{2c_2}) \\
	\sqrt{\frac{c_1\phi_1 {A}_1}{2}} + \sqrt{\frac{c_2\phi_2({A}-{A}_1)}{2}}, &\text{if } {A}_1 \in [{A} - \frac{\phi_2}{2c_2},\frac{\phi_1}{2c_1}) \\
	\phi_1 + \sqrt{\frac{c_2\phi_2({A}-{A}_1)}{2}}, &\text{if } {A}_1 \in [\frac{\phi_1}{2c_1},{A}]
	\end{cases}
	\end{equation}
	It is strictly increasing (decreasing) on the first (third) interval, taking the maximum value of $\phi_2 + Q_1$ ($\phi_1 + Q_2$). The value in the middle interval is upper bounded by $\sqrt{\frac{{A}(c_1\phi_1+c_2\phi_2)}{2}} \leq \frac{\phi_1+\phi_2}{2}$. The objective will never take the maximal value in the middle interval because $\phi_i + Q_{-i} > \frac{\phi_1+\phi_2}{2}$ for some $i=1,2$.

	\noindent\textbf{Case 4:} The objective is written as
	\begin{equation}
	W_\text{d}({A}_1) = 
	\begin{cases}
	\sqrt{\frac{c_1\phi_1 {A}_1}{2}} + \phi_2, &\text{if } {A}_1 \in [0,\frac{\phi_1}{2c_1}) \\
	\phi_1+\phi_2, &\text{if } {A}_1 \in [\frac{\phi_1}{2c_1},{A} - \frac{\phi_2}{2c_2}) \\
	\phi_1 + \sqrt{\frac{c_2\phi_2({A}-{A}_1)}{2}}, &\text{if } {A}_1 \in [{A} - \frac{\phi_2}{2c_2},{A}]
	\end{cases}
	\end{equation}
	which is maximized for any ${A}_1 \in [\frac{\phi_1}{2c_1},{A} - \frac{\phi_2}{2c_2})$.
\end{proof}

The result below completely characterizes the commander's optimal Stage 1 assignment in $\text{CAP}_\text{d}$ under the per-unit cost setting. 

\begin{proposition}\label{prop:commander_det}
	Consider the deterministic commander assignment problem $\text{CAP}_\text{d}$. Denote $j=\argmin{i\in\{1,2\}} \frac{\phi_i}{2c_i}$ and $k = \argmax{i\in{1,2}} c_i$. We define the intervals $I_\ell := c_{-k}\cdot[1-\frac{\sqrt{3}}{2}, 1+\frac{\sqrt{3}}{2}  ]$ and $I_r := c_k \cdot[1-\frac{\sqrt{3}}{2}, 1+\frac{\sqrt{3}}{2}]$. Enumerate the following four statements.
	\begin{enumerate}[label=(\roman*)]
		\item For $c \notin I_\ell \cup I_r$, we have $W_d^* = \frac{c_1\phi_1+c_2\phi_2}{8c}$, $A_j^* = \frac{c_j\phi_j}{8c^2}$, and the expenditure is $A^* = \frac{c_1\phi_1+c_2\phi_2}{8c^2}$.
		\item For $c \in I_r \backslash I_\ell$ and $j= k$, or $c \in I_\ell \backslash I_r$ and $j\neq  k$, we have $W_d^* = \frac{c_{-j}\phi_{-j}}{8c} + (1-\frac{c}{2c_j})\phi_j$, $A_j^* = \frac{\phi_j}{2c_j}$, and $A^* = \frac{c_{-j}\phi_{-j}}{8c^2} + \frac{\phi_j}{2c_j}$.
		\item For $c \in I_\ell \backslash I_r$ and $j= k$, or $c \in I_r \backslash I_\ell$ and $j\neq k$, we have $W_d^* = \frac{c_j\phi_j}{8c} + (1-\frac{c}{2c_{-j}})\phi_{-j}$, $A_j^* = \frac{c_j\phi_j}{8c^2}$, and $A^* = \frac{c_j\phi_j}{8c^2} + \frac{\phi_{-j}}{2c_{-j}}$. 
		\item For $c \in I_\ell \cap I_r$, we have $W_d^* = (1-\frac{c}{2c_1})\phi_1 + (1-\frac{c}{2c_2})\phi_2$, $A_j^* = \frac{\phi_j}{2c_j}$, and $A^* = \frac{\phi_1}{2c_1}+\frac{\phi_2}{2c_2}$.
	\end{enumerate}
	Denote $s_1 = \frac{1}{2}\sqrt{\frac{c_j}{\phi_j}(c_1\phi_1+c_2\phi_2) }$ and $s_2 = \frac{c_{-j}}{2}\sqrt{\frac{c_j\phi_{-j}}{c_j\phi_{-j} - c_{-j}\phi_j}}$. Then the optimal commander assignment is given as follows. 
	\begin{itemize}[leftmargin=*]
		\item Suppose $c < \frac{c_j}{2}$. Then $W_d^* = (1-\frac{c}{2c_1})\phi_1 + (1-\frac{c}{2c_2})\phi_2$, $A_j^* = \frac{\phi_j}{2c_j}$, and $A^* = \frac{\phi_1}{2c_1}+\frac{\phi_2}{2c_2}$.
		\item Suppose $\frac{c_j}{2} \leq c < \min\{s_1,s_2\}$. Then the solution follows from whether the condition of (iii) or (iv) holds.
		\item Suppose $\min\{s_1,s_2\} \leq c < \max\{s_1,s_2\}$. If $s_1 \leq s_2$, then the solution follows from whether the condition of (i), (iii), or (iv) holds. If $s_1 > s_2$, then the solution follows from whether the condition of (i), (ii), or (iv) holds.
		\item Suppose $\max\{s_1,s_2\} \leq c$. Then the solution follows from whether the condition of (i), (ii), (iii), or (iv) holds.
	\end{itemize}
\end{proposition}

\begin{proof}[Proof of Proposition \ref{prop:commander_det}]
	The commander's optimal assignment in $\text{CAP}_\text{d}$ under the per-unit cost setting follows from solving the optimization problem
	\begin{equation}
	\max_{{A} \geq 0} \{ W_\text{d}^*({A}) - c{A} \}
	\end{equation}
	where we denote $W_\text{d}^*({A})$ as the commander's optimal payoff given a fixed use-it-or-lose-it budget $A$, characterized from Lemma \ref{lem:commander_det_nocost}. In general, the objective above is not concave and there are at most three points of discontinuity. There may exist up to four critical points in ${A} \in [0,\infty)$, depending on the value of $c$. These points are indicated by ${A}^*$ in the enumerated list (i) - (iv) from the Proposition statement.  A critical point exists on the interval ${A} \in (0, \min_{i=1,2}\{\frac{\phi_i}{2c_i}\})$ if $c>\frac{1}{2}\sqrt{\frac{c_j}{\phi_j}(c_1\phi_1 + c_2\phi_2)}$, on the interval ${A} \in [\frac{\phi_j}{2c_j}, \frac{\phi_{-j}}{2c_{-j}})$ if $c>\frac{c_{-j}}{2}\sqrt{\frac{c_j\phi_{-j}}{c_j\phi_{-j} - c_{-j}\phi_j }}$, and on the interval ${A} \in [\frac{\phi_{-j}}{2c_{-j}}, \frac{\phi_1}{2c_1} + \frac{\phi_2}{2c_2})$ if $c > \frac{c_j}{2}$. A critical point always exists at ${A} = \frac{\phi_1}{2c_1} + \frac{\phi_2}{2c_2}$, as 0 is always contained in the sub-differential. The largest critical point is determined by the conditions listed as (i) to (iv) in the statement of Proposition \ref{prop:commander_det}. 
\end{proof}

We remark in this result that the Stage 1 objective function for the commander is generally non-concave, discontinuous, and can have up to four critical points depending on the values of $c$, $c_i$, and $\phi_i$. We note that the most amount of resources the commander will invest is $\frac{\phi_1}{2c_1}+\frac{\phi_2}{2c_2}$, which occurs for low costs $c < \frac{c_j}{2}$. Indeed, the quantity $\frac{\phi_i}{2c_i}$ is the amount of resources needed to win competition $i$ outright.

\section{Optimal randomized assignments -- per-unit cost setting}\label{sec:randomized_proofs}

Analogously to Appendix \ref{sec:deterministic_proofs}, we proceed to derive the solution of $\text{CAP}$ under the per-unit cost setting by first solving the optimal opponent investment decision in Stage 2. Recall that by the beginning of Stage 2, the commander has chosen a distribution $\Prob \in \mcal{P}$ on assignments. The marginal distribution $\Prob_i$ is thus represented by a Bernoulli distribution $(A_i^h,A_i^\ell,p_i)$. The Stage 2 decision problem for opponent $\mcb_i$ is therefore
\begin{equation}\label{eq:Bi_CAP_problem}
    \max_{B_i \geq 0} \left\{\phi_i\cdot\pi_\mcb(\Prob_i,B_i) - c_i B_i \right\}
\end{equation}
where $\pi_\mcb$ is given in Theorem \ref{thm:BNE_regions}.

\begin{lemma}\label{lem:Bstar_CAP}
    Consider the Stage 2 decision problem for opponent $\mcb_i$ in $\text{CAP}$ \eqref{eq:Bi_CAP_problem}.
     If $\frac{A_i^\ell}{A_i^h} \geq \frac{1-p_i}{2-p_i}$, then the optimal investment for $\mcb_i$ is given by \eqref{eq:Bstar_CAPd}, i.e.
    \begin{equation}
        B_i^*(\Prob_i) = 
        \begin{cases}
            \sqrt{\frac{\phi_i \bar{A}_i}{2 c_i}}, &\text{if } c_i < \frac{\phi_i}{2\bar{A}_i} \\
            0, &\text{otherwise}
        \end{cases}
    \end{equation}
    where $\bar{A}_i := p_iA_i^h + (1-p_i)A_i^\ell$. If $\frac{A_i^\ell}{A_i^h} < \frac{1-p_i}{2-p_i}$, then
    \begin{equation}\label{eq:Bstar_CAP}
		B_i^*(\Prob_i) = 
		\begin{cases}
			\sqrt{\frac{\bar{A}_i \phi_i}{2c_i}}, &\text{if } c_i \in [0, \lambda_i) \\
			\sqrt{\frac{(1-p_i)A_i^\ell\phi_i}{2c_i}}, &\text{if } c_i \in [\lambda_i,\frac{(1-p_i)\phi_i}{2A_i^\ell}) \\
			0, &\text{if } c_i \geq \frac{(1-p_i)\phi_i}{2A_i^\ell}
		\end{cases}
	\end{equation}
	where we have defined $\lambda_i \triangleq \frac{\phi_i}{2(A_i^h)^2}(\sqrt{(1-p_i)A_i^\ell} + \sqrt{\bar{A}_i})^2$. 
\end{lemma}
\begin{proof}

	The optimal investment is the maximizer of the optimization problem \eqref{eq:Bi_CAP_problem}. For simpler exposition, we will drop the subscript $i$ on all variables in this proof since identical analysis applies to both opponents. 
	
	First, suppose  $\frac{A^\ell}{A^h} \geq \frac{1-p}{2-p}$. From \eqref{eq:piA}, we have $\pi_\mcb(\Prob_\mca,B) = \pi_\mcb^\text{CI}(\bar{A},B)$ for all $B' \geq 0$, and therefore $B^*$ coincides with \eqref{eq:Bstar_CAPd}.
	
	Now, suppose $\frac{A^\ell}{A^h} < \frac{1-p}{2-p}$. From Theorem \ref{thm:BNE_regions}, we can write $\pi_\mcb(\Prob_\mca,B)$ as
	\begin{equation}\label{eq:piB_XB}
		\phi \cdot
		\begin{cases}
			(1-p)\frac{B'}{2 A^\ell}, &\text{if } B' \leq A^\ell \\
			(1-p)\left(1 - \frac{A^\ell}{2 B} \right), &\text{if } A^\ell < B \leq Y^\ell \\
			(1-p)(1-\frac{A^\ell}{A^h}) - \frac{\sqrt{\bar{A}(\bar{A} - pA^h)}}{A^h} + \lambda B, &\text{if }  Y^\ell < B \leq Y^h \\
			p\left(1 - \frac{A^h}{2B}\right) + (1-p)\left(1 - \frac{A^\ell}{2B'}\right), &\text{if } Y^h < B
		\end{cases}
	\end{equation}
	where we have written $Y^\ell \triangleq \frac{A^h \sqrt{(1-p)A^\ell}}{\sqrt{(1-p)A^\ell} + \sqrt{\bar{A}}}$, and $Y^h \triangleq \frac{A^h \sqrt{\bar{A}}}{\sqrt{(1-p)A^\ell} + \sqrt{\bar{A}}}$. The entries in the expression above correspond to payoffs in regions $\mcal{R}_3$, $\mcal{R}_4$, $\mcal{R}_5$, and $\mcal{R}_1$, respectively (Theorem \ref{thm:BNE_regions}). The values for $B^*$ result from solving $\frac{\partial}{\partial B}\left(\pi_\mcb(\Prob_\mca,B)-cB\right) = 0$. Note that the first and third entries provide linear returns on investment $B$. Therefore, there are multiple maximizers when the per-unit cost coincides with these slopes. In particular, when $c = \lambda$ or $c=\frac{1-p}{2A^\ell}$, the payoff $\pi_\mcb(\Prob_\mca,B) - c\cdot B$ is constant for all $B \in [Y^\ell,Y^h]$ or $B\in [0,A^\ell]$, respectively. Since player $\mcb$ prefers the lowest investment level, we thus obtain \eqref{eq:Bstar_CAP}.
\end{proof}

Now that we have established the optimal decision in Stage 2, we can now consider the optimal commander's assignment decision problem in Stage 1, which can be stated as:

\begin{equation}\label{eq:CAP_stage1_problem}
    W^* \triangleq \max_{\Prob \in \mcal{P}} \left\{\phi_1\cdot\pi_\mca(\Prob_1,B_1^*(\Prob_1)) + \phi_2\cdot\pi_\mca(\Prob_2,B_2^*(\Prob_2)) - c\cdot\E_\Prob[A_1 + A_2]\right\}
\end{equation}

Before addressing \eqref{eq:CAP_stage1_problem} fully, we first consider an intermediate step where the commander has a fixed budget $A>0$ and zero cost $c = 0$. The intermediate decision problem is:
\begin{equation}\label{eq:zerocost_randomized_commander}
    W_A^* \triangleq \max_{\Prob \in \mcal{P}(A)} \left\{\phi_1\cdot\pi_\mca(\Prob_1,B_1^*(\Prob_1)) + \phi_2\cdot\pi_\mca(\Prob_2,B_2^*(\Prob_2)) \right\}
\end{equation}
where $B_i^*$ is given in Lemma \ref{lem:Bstar_CAP}. Note that any feasible $\Prob \in \mcal{P}({A})$ induces expected endowments $\bar{A}_{i} = p_i A_i^h + (1-p_i)A_i^\ell$ for each sub-colonel, where $\bar{A}_1 + \bar{A}_2 = A$ is satisfied. The optimization \eqref{eq:zerocost_randomized_commander} can thus be re-written as:
\begin{equation}\label{eq:alt_CAP_problem}
    \max_{\substack{ (\bar{A}_1,\bar{A}_2) : \\ \bar{A}_1+\bar{A}_2 = A}} \left\{ \sum_{i=1,2} \max_{\substack{\Prob_i = (A_i^h,A_i^\ell,p_i) : \\ p_iA_i^h + (1-p_i)A_i^\ell = \bar{A}_i } }\phi_i\cdot \pi_\mca(\Prob_i,B_i^*(\Prob_i))  \right\}
\end{equation}

The following Lemma provides the optimal Bernoulli distribution for sub-colonel $i$ given a fixed expected endowment $\bar{A}_i$, i.e. the solution of each inner maximization above.

\begin{lemma}\label{lem:opt_budget_randomization}
    Given a fixed expected endowment $\bar{A}_i > 0$ for sub-colonel $i$, 
    \begin{equation}\label{eq:opt_obfuscation_problem}
        \begin{aligned}
            \Pi_i^*(\bar{A}_i;\phi_i) &\triangleq \max_{\substack{\Prob_i = (A_i^h,A_i^\ell,p_i) : \\ p_iA_i^h + (1-p_i)A_i^\ell = \bar{A}_i } }\phi_i\cdot\pi_\mca(\Prob_i,B_i^*(\Prob_i)) \\
            &= 
            \begin{cases}
    			\sqrt{2c_i\phi_i\bar{A}_i} &\text{if } c_i \in (0,\frac{\phi_i}{2\bar{A}_i}) \\
    			\phi_i, &\text{if } c_i \geq \frac{\phi_i}{2\bar{A}_i}
    		\end{cases}
        \end{aligned}
    \end{equation}
    The optimal Bernoulli distribution  that achieves the above payoff is 
    \begin{equation}\label{eq:opt_obfuscation}
		(A_i^{h*},A_i^{\ell*},p_i^*) = 
		\begin{cases}
			\left(\sqrt{\frac{\bar{A}_i\phi_i}{2c_i}},0,\sqrt{\frac{2c_i\bar{A}_i}{\phi_i}}\right) &\text{if } c \in (0,\frac{\phi_i}{2\bar{A}_i}) \\
			\left(\bar{A}_i,\bar{A}_i,\times \right), &\text{if } c_i \geq \frac{\phi_i}{2\bar{A}_i}
		\end{cases}
	\end{equation}
\end{lemma}
The optimal Bernoulli assignment can double sub-colonel $i$'s payoff compared to the deterministic assignment $\bar{A}_i$. 

\begin{proof}
    We first characterize the solution of finding the optimal Bernoulli distribution given a fixed $p_i \in [0,1]$ and expected endowment $\bar{A}_i$:
    \begin{equation}\label{eq:opt_obfuscation_problem_nop}
        \Pi_i^*(\bar{A}_i,p_i;\phi_i) \triangleq \max_{\substack{(A_i^h,A_i^\ell) : \\ p_iA_i^h + (1-p_i)A_i^\ell = \bar{A}_i } }\phi_i\cdot\pi_\mca(\Prob_i,B_i^*(\Prob_i))
    \end{equation}
    and then optimize $\Pi_i^*(\bar{A}_i,p_i;\phi_i)$ over $p_i$ to obtain the solution of \eqref{eq:opt_obfuscation_problem}.
    
    Let us define $\Pi_i(\Prob_i)\triangleq\phi\cdot\pi_A(\Prob_i,B_i^*(\Prob_i))$ and $f = \frac{A^\ell}{A^h} \in [0,1]$ as a change of variable. Using Lemma \ref{lem:Bstar_CAP}, if $f < \frac{1-p}{2-p}$, we can write
    \begin{equation}\label{eq:piA_c}
    	\Pi_i(\Prob_i) =  
    	\begin{cases}
        	\sqrt{\frac{c\phi\bar{A}}{2}}, &\text{if } c \in [0,\lambda) \\
        	p\phi + \sqrt{\frac{c\phi(1-p)A^\ell}{2}}, &\text{if } c \in [\lambda, \frac{(1-p)\phi}{2A_2}) \\
        	\phi, &\text{if } c \in [\frac{(1-p)\phi}{2A_2},\infty)
    	\end{cases}
	\end{equation}
    where $\lambda = \phi\frac{p + (1-p)f}{2\bar{A}}(\sqrt{(1-p)f} + \sqrt{p + (1-p)f})^2$. If $f \geq \frac{1-p}{2-p}$, then
    \begin{equation}\label{eq:piA_c_CI}
    	\Pi_i(\Prob_i) =  
    	\begin{cases}
        	\sqrt{\frac{c\phi\bar{A}}{2}}, &\text{if } c < \frac{\phi}{2\bar{A}} \\
        	\phi, &\text{else}
    	\end{cases}
	\end{equation}

    First, we observe that $\lambda(f)$ is strictly increasing in $f$, taking the value $\frac{p^2\phi}{2\bar{A}}$ for $f= 0$. From \eqref{eq:piA_c}, we then have $\Pi_i(\Prob_i)=\sqrt{\frac{c\phi\bar{A}}{2}}$ for all $c <  \frac{p^2\phi}{2\bar{A}}$ regardless of $f$. Therefore, $\Pi_i^*(\bar{A},p;\phi) = \sqrt{\frac{c\phi\bar{A}}{2}}$ for $c <  \frac{p^2\phi}{2\bar{A}}$. When $c \geq \frac{\phi}{2\bar{A}}$, we observe $f$ can be set to 1, making \eqref{eq:piA_c_CI} active, and ensuring the maximum payoff of $\phi$. 
	
	For $c \in [\frac{p^2\phi}{2\bar{A}},\frac{\phi}{2\bar{A}})$, we claim that the $f^*$ that satisfies $\lambda(f^*) = c$ characterizes the solution $A^{h*},A^{\ell*}$ to \eqref{eq:opt_obfuscation_problem_nop}. One can solve this equation for $f^*$ as follows: make the substitution $y = p + (1-p)f$ to obtain $y(\sqrt{y-p} + \sqrt{y})^2 = 2c\bar{A}/\phi$, which has the solution $y = \frac{2c\bar{A}/\phi}{2\sqrt{2c\bar{A}/\phi} - p}$. We then get $f^* = \frac{2c\bar{A}/\phi}{(1-p)(2\sqrt{2c\bar{A}/\phi} - p)} - \frac{p}{1-p}$ and subsequent endowments $A^{h*} = \phi\frac{2\sqrt{2c\bar{A}/\phi} - p}{2c}$ and $A^{\ell*} =  \frac{1}{1-p}\left( \bar{A} - \frac{p\phi(2\sqrt{2c\bar{A}/\phi}-p) }{2c}\right)$. Denoting $\Prob_i^* = (A^{h*},A^{\ell*},p)$, from \eqref{eq:piA_c} we have $\Pi_i(\Prob_i^*) = p\phi + \sqrt{\frac{c\phi(1-p)A^{\ell*}}{2}} = \frac{p\phi}{2} + \sqrt{\frac{c\phi\bar{A}}{2}}$. The second equality follows due to $c \geq \frac{p^2\phi}{2\bar{A}}$.
	
	We verify that $\Pi_i(\Prob_i^*) > \Pi_i(\Prob_i)$ for any other $\Prob_i = (A^h,A^\ell,p)$ satisfying the expected endowment constraint. Let $f = A^\ell/A^h$. For any $f > f^*$, $\lambda(f) > \lambda(f^*)$.  Since $c$ is fixed, we thus obtain
	\begin{equation}
	    \Pi_i(\Prob_i) = \sqrt{\frac{c\phi\bar{A}}{2}}  < \Pi_i(\Prob_i^*)
	\end{equation}
	For any $f < f^*$, we must have $A^\ell < A^{\ell*}$ (and $A^h > A^{h*}$). We then have $\Pi_i(\Prob_i) = p\phi + \sqrt{\frac{c\phi(1-p)A^\ell}{2}} < p\phi + \sqrt{\frac{c\phi(1-p)A^{\ell*}}{2}} = \Pi_i(\Prob_i^*)$.
	
	We thus obtain 
	\begin{equation}\label{eq:payoff_nop}
    	\Pi_i^*(\bar{A},p;\phi) =
    	\begin{cases}
    		\sqrt{\frac{c\phi\bar{A}}{2}}, &\text{if } c < \frac{p^2\phi}{2\bar{A}} \\
    		\frac{p\phi}{2} + \sqrt{\frac{c\phi\bar{A}}{2}}, &\text{if }  \frac{p^2\phi}{2\bar{A}} \leq c < \frac{\phi}{2\bar{A}} \\
    		\phi, &\text{if } \frac{\phi}{2\bar{A}} \leq c
    	\end{cases} 
	\end{equation}
	and the optimal randomization is
	\begin{equation}
		(A^{h*},A^{\ell*}) = \left(\phi\frac{2\sqrt{2c\bar{A}/\phi} - p}{2c},\frac{1}{1-p}\left( \bar{A} - \frac{p\phi(2\sqrt{2c\bar{A}/\phi}-p) }{2c}\right) \right)
	\end{equation}
	We can now readily obtain the optimal value of \eqref{eq:opt_obfuscation_problem}. The optimization \eqref{eq:opt_obfuscation_problem} is equivalent to maximizing $\Pi_i^*$ \eqref{eq:payoff_nop} over $p \in [0,1]$. We observe sub-colonel $i$ obtains the maximum payoff $\phi$ for high costs $c \geq \frac{\phi}{2\bar{A}}$, irrespective of $p$. So, suppose $c < \frac{\phi}{2\bar{A}}$ is fixed.  The expression $\frac{p\phi}{2} + \sqrt{\frac{c\phi\bar{A}}{2}}$ is strictly increasing on $p \in [0,\sqrt{2c\bar{A}/\phi}]$.	This payoff is attainable only for $p \in [0,\sqrt{2c\bar{A}/\phi}]$, and hence $\Pi_i^*$ is maximized at $p^* = \sqrt{2c\bar{A}/\phi}$. This gives the payoff $\sqrt{2c\phi\bar{A}} = 2\cdot\pi_\mca^\text{CI}$.
\end{proof}

With the optimal marginal distributions established with the above Lemma, we are now ready to derive the solution of \eqref{eq:zerocost_randomized_commander}: the commander's optimal assignment given a fixed budget $A > 0$ and zero cost $c= 0$. Below, we provide the optimal assignment $\Prob \in \mcal{P}(A)$ and the resulting payoff $W_A^*$.

\begin{lemma}\label{lem:commander_rand_nocost}
	Suppose $c = 0$, ${A} < \infty$. The optimal assignment $\Prob^*$ that solves \eqref{eq:zerocost_randomized_commander} is given in the following cases below.
	
	\noindent\textbf{Case 1:} Suppose ${A} < \min_{i=1,2} \frac{c_1\phi_1+c_2\phi_2}{2c_i^2}$. Let ${A}_1^* = \frac{c_1\phi_1 {A}}{c_1\phi_1 + c_2\phi_2}$, ${A}_2^* = \frac{c_2\phi_2 {A}}{c_1\phi_1 + c_2\phi_2}$, $p_1^* = \sqrt{2c_1{A}_1^*/\phi_1}$, and $p_2^* = \sqrt{2c_2{A}_2^*/\phi_2}$. Then
	\begin{equation}
    	\begin{aligned}
    	&\Prob^*(0,0) = (1-p_1^*)(1-p_2^*), \quad &\Prob^*\left(0,\sqrt{\frac{{A}_2^*\phi_2}{2c_2}} \right) = (1-p_1^*) p_2^* \\
    	&\Prob^*\left(\sqrt{\frac{{A}_1^*\phi_1}{2c_1}},0\right) = p_1^*(1-p_2^*), \quad &\Prob^*\left(\sqrt{\frac{{A}_1^*\phi_1}{2c_1}},\sqrt{\frac{{A}_2^*\phi_2}{2c_2}} \right) = p_1^*p_2^*
    	\end{aligned}
	\end{equation}
	The resulting performance is $W_A^*=\sqrt{2{A}(c_1\phi_1 + c_2\phi_2)}$.
	
	\noindent\textbf{Case 2:} Suppose $\min_{i=1,2} \frac{c_1\phi_1+c_2\phi_2}{2c_i^2} \leq {A} < \frac{\phi_1}{2c_1}+\frac{\phi_2}{2c_2}$. Let $k = \argmax{i=1,2} c_i$, ${A}_k^* = \frac{\phi_k}{2c_k}$, ${A}_{-k}^* = {A} - \frac{\phi_k}{2c_k}$, and $p_{-k}^* = \sqrt{2c_{-k}{A}_{-k}^*/\phi_{-k}}$. Then
	\begin{equation}
    	\begin{aligned}
    	\Prob_k^*({A}_k^*) = 1 \quad \text{and} \quad
    	\Prob_{-k}^* = \left(\sqrt{\frac{{A}_{-k}^*\phi_{-k}}{2c_{-k}}},0,p_{-k}^* \right).
    	\end{aligned}
	\end{equation}
	The resulting performance is $W_A^* = \phi_k + \sqrt{2c_{-k}\phi_{-k}\left({A} - \frac{\phi_k}{2c_k}\right)}$.

	\noindent\textbf{Case 3:} Suppose ${A} \geq \frac{\phi_1}{2c_1} + \frac{\phi_2}{2c_2}$. Then any $\Prob^*$ that satisfies $\Prob^*({A}_1^*,{A}_2^*) = 1$ for some ${A}_1^* \in [\frac{\phi_1}{2c_1},{A}-\frac{\phi_2}{2c_2}]$ and ${A}_2^* = {A} - {A}_1^*$ is an optimal assignment. The resulting performance is $W_A^*=\phi_1 + \phi_2$.
\end{lemma}
\begin{proof}
	From \eqref{eq:alt_CAP_problem} and Lemma \ref{lem:opt_budget_randomization},
	the optimization \eqref{eq:zerocost_randomized_commander} is equivalent to solving
	\begin{equation}
    	\max_{\bar{A}_1 \in [0,{A}]} \{  \Pi_1^*(\bar{A}_1 ; \phi_1) + \Pi_2^*({A} - \bar{A}_1 ; \phi_2) \}
	\end{equation}
	where $\Pi_i^*$ is given in \eqref{eq:opt_obfuscation_problem}. Let us denote the objective function in the brackets above as $W_A(\bar{A}_1)$. It takes distinct forms depending on the values $\{c_i,\phi_i\}_{i=1,2}$:
	
	\noindent\textbf{(a)} $c_1 < \frac{\phi_1}{2{A}}$ and $c_2 < \frac{\phi_2}{2{A}}$. The objective is written as $\sqrt{2c_1\phi_1{A}_1} + \sqrt{2c_2\phi_2({A} - {A}_1)}, \ \forall {A}_1 \in [0,{A}]$.    This function is concave. Setting the derivative (w.r.t. ${A}_1$) to zero, we obtain ${A}_1^* = \frac{c_1\phi_1}{c_1\phi_1 + c_2\phi_2}{A}$ and the optimal payoff is $\sqrt{2{A}(c_1\phi_1 + c_2\phi_2)}$. Here, the distribution $P^*$ corresponding to \textbf{Case 1} in the Theorem statement induces the randomizations given by \eqref{eq:opt_obfuscation} for each sub-colonel. 
	
	\noindent\textbf{(b)} ${A} < \frac{\phi_1}{2c_1}$ and ${A} \geq \frac{\phi_2}{2c_2}$. The objective is written as
	\begin{equation}\label{eq:case_b}
	W_A({A}_1)=
	\begin{cases}
	\sqrt{2c_1\phi_1{A}_1} + \phi_2, &\text{if } {A}_1 \in [0,{A} - \frac{\phi_2}{2c_2}) \\
	\sqrt{2c_1\phi_1{A}_1} + \sqrt{2c_2\phi_2({A} - {A}_1)}, &\text{if } {A}_1 \in [{A} - \frac{\phi_2}{2c_2},{A}] \\
	\end{cases}
	\end{equation}
	This function is concave and continuous on $[0,{A}]$. It is strictly increasing on the first interval. It is strictly decreasing for ${A}_1 > \frac{c_1\phi_1 {A}}{c_1\phi_1 + c_2\phi_2}$. The maximizer is given by
	\begin{equation}
	{A}_1^* =
	\begin{cases}
	{A} - \frac{\phi_2}{2c_2}, &\text{if } c_2 > r_2(c_1) \\
	\frac{c_1\phi_1}{c_1\phi_1 + c_2\phi_2}{A}, &\text{if } c_2 \leq r_2(c_1) \\
	\end{cases}
	\end{equation}
	where we define $r_i(c_{-i}) = \frac{1}{2}\left[\frac{\phi_i}{2{A}} + \sqrt{\left(\frac{\phi_i}{2{A}} \right)^2 + \frac{2c_{-i}\phi_{-i}}{{A}}} \right]$ for $i=1,2$. The condition $c_2 \leq r_2(c_1)$ above arises when the derivative of the second entry in \eqref{eq:case_b} is positive at ${A}_1 = {A} - \frac{\phi_2}{2c_2}$. Hence, when $c_2 > r_2(c_1)$, ${A}_1^* = {A} - \frac{\phi_2}{2c_2}$. Therefore, if $c_2 \leq r_2(c_1)$, the solution $\Prob^*$ corresponds to \textbf{Case 1}. If $c_2 > r_2(c_1)$, the solution $\Prob^*$ corresponds to \textbf{Case 2} ($k=2$).

	\noindent\textbf{(c)} $c_1 \geq \frac{\phi_1}{2{A}}$ and $c_2 < \frac{\phi_2}{2{A}}$. The objective is written as
	\begin{equation}
	W_A({A}_1)=
	\begin{cases}
	\sqrt{2c_1\phi_1{A}_1} + \sqrt{2c_2\phi_2({A} - {A}_1)}, &\text{if } {A}_1 \in [0,\frac{\phi_1}{2c_1}] \\
	\phi_1 + \sqrt{2c_2\phi_2({A}-{A}_1)}, &\text{if } {A}_1 \in (\frac{\phi_1}{2c_1},{A}] \\
	\end{cases}
	\end{equation}
	Analogous arguments from \textbf{(b)} yield 
	\begin{equation}
	{A}_1^* =
	\begin{cases}
	\frac{\phi_1}{2c_1}, &\text{if } c_1 > r_1(c_2) \\
	\frac{c_1\phi_1}{c_1\phi_1 + c_2\phi_2}{A}, &\text{if } c_1 \leq r_1(c_2)  \\
	\end{cases}
	\end{equation}
	Here, the optimal distribution $\Prob^*$ corresponds to \textbf{Case 1} if $c_1 \leq r_1(c_2)$ and \textbf{Case 2} ($k=1$) if $c_1 > r_1(c_2)$.
	
	\noindent\textbf{(d)} $c_1 \geq \frac{\phi_1}{2{A}}$ and $c_2 \geq \frac{\phi_2}{2{A}}$. First, we consider the sub-case ${A} - \frac{\phi_2}{2c_2} < \frac{\phi_1}{2c_1}$. The objective is written as
	\begin{equation}
	W_A({A}_1)=
	\begin{cases}
	\sqrt{2c_1\phi_1{A}_1} + \phi_2, &\text{if } {A}_1 \in [0,{A} - \frac{\phi_2}{2c_2}) \\
	\sqrt{2c_1\phi_1{A}_1} + \sqrt{2c_2\phi_2({A} - {A}_1)}, &\text{if } {A}_1 \in [{A} - \frac{\phi_2}{2c_2},\frac{\phi_1}{2c_1}] \\
	\phi_1 + \sqrt{2c_2\phi_2({A}-{A}_1)}, &\text{if } {A}_1 \in (\frac{\phi_1}{2c_1},{A}] \\
	\end{cases}
	\end{equation}
	This function is concave, strictly increasing on $[0,{A} - \frac{\phi_2}{2c_2})$, and strictly decreasing on $(\frac{\phi_1}{2c_1},{A}]$. The maximizer is given by
	\begin{equation}
	{A}_1^* =
	\begin{cases}
	{A} - \frac{\phi_2}{2c_2}, &\text{if } c_2 > r_2(c_1) \\
	\frac{c_1\phi_1}{c_1\phi_1 + c_2\phi_2}{A}, &\text{if } c_2 \leq r_2(c_1) \text{ and } c_1 \leq r_1(c_2) \\
	\frac{\phi_1}{2c_1}, &\text{if } c_1 > r_1(c_2)
	\end{cases}
	\end{equation}
	Thus, the optimal distribution $P^*$ corresponds to \textbf{Case 2} ($k=2$) if $c_2 > r_2(c_1)$, to \textbf{Case 1} if $c_2 \leq r_2(c_1)$ and $c_1 \leq r_1(c_2)$, and to \textbf{Case 2} ($k=1$) if $c_1 > r_1(c_2)$.
	
	Lastly, we consider the sub-case ${A} - \frac{\phi_2}{2c_2} \geq \frac{\phi_1}{2c_1}$. The objective can be written as
	\begin{equation}
	W({A}_1)=
	\begin{cases}
	\sqrt{2c_1\phi_1{A}_1} + \phi_2, &\text{if } {A}_1 \in [0,\frac{\phi_1}{2c_1}) \\
	\phi_1 + \phi_2, &\text{if } {A}_1 \in [\frac{\phi_1}{2c_1},{A} - \frac{\phi_2}{2c_2}] \\
	\phi_1 + \sqrt{2c_2\phi_2({A}-{A}_1)}, &\text{if } {A}_1 \in ({A} - \frac{\phi_2}{2c_2},{A}] \\
	\end{cases}
	\end{equation}
	The optimal distribution $\Prob^*$ here corresponds to \textbf{Case 3}.
\end{proof}

The result below gives the solution to the optimal Stage 1 assignment problem \eqref{eq:CAP_stage1_problem} and corresponding final payoff $W^*$ for the commander in the per-unit costs setting. 

\begin{proposition}\label{prop:commander_costs}
    The commander's optimal Stage 1 assignment $\Prob^*$ and corresponding final payoff $W^*$ in the extensive-form game $\text{CAP}$ under the per-unit cost setting is given as follows. Let $k = \argmax{i=1,2} c_i$.
	
	
	\noindent$\bullet$ If $c \leq c_{-k}$, then $W^* = (1-\frac{c}{2c_1})\phi_1 + (1-\frac{c}{2c_2})\phi_2$ and $\Prob^*(\frac{\phi_1}{2c_1},\frac{\phi_2}{2c_2}) = 1$.
	The resource expenditure is ${A}^* = \frac{\phi_1}{2c_1} + \frac{\phi_2}{2c_2}$.
	
	\noindent$\bullet$ If $c_{-k} < c < c_k$, then $W^*=\phi_k + \frac{c_{-k}}{c}\phi_{-k}$ and
	\begin{equation}
	\begin{aligned}
	\Prob_k^*(\frac{\phi_k}{2c_k}) = 1 \quad \text{and} \quad
	\Prob_{-k}^* = \left(\frac{\phi_{-k}}{2c}, 0, \frac{c_{-k}}{c} \right).
	\end{aligned}
	\end{equation}
	The expected resource expenditure is ${A}^* = \frac{c_{-k}\phi_{-k}}{2c^2} + \frac{\phi_k}{2c_k}$.
	
	\noindent$\bullet$ If $c \geq c_k$, then $W^*=\frac{c_1\phi_1+c_2\phi_2}{2c}$ and \begin{equation}\label{eq:maximal_randomization}
	\begin{aligned}
    	&\Prob^*(0,0) = (1-c_1/c)(1-c_2/c), \quad &\Prob^*\left(0,\frac{\phi_2}{2c} \right) = (1-c_1/c)\cdot (c_2/c) \\
    	&\Prob^*\left(\frac{\phi_1}{2c},0\right) = (c_1/c)\cdot (1-c_2/c),  &\Prob^*\left(\frac{\phi_1}{2c},\frac{\phi_2}{2c} \right) = \frac{c_1c_2}{c^2}
	\end{aligned}
	\end{equation}
	The expected resource expenditure is ${A}^* = \frac{c_1\phi_1+c_2\phi_2}{2c^2}$.
\end{proposition}

\begin{proof}[Proof of Proposition \ref{prop:commander_costs}]
	The solution of \eqref{eq:CAP_stage1_problem} under per-unit costs follows from solving the optimization problem
	\begin{equation}
	\max_{{A} \geq 0} \{ W_A^* - c\cdot{A} \}
	\end{equation}
	where we denote $W_A^*$ as the performance from Lemma \ref{lem:commander_rand_nocost}. It is a concave objective, and the critical point ${A}^*$ lies in $(0,\min_{i=1,2} \frac{c_1\phi_1+c_2\phi_2}{2c_i^2})$, $[\min_{i=1,2} \frac{c_1\phi_1+c_2\phi_2}{2c_i^2}, \frac{\phi_1}{2c_1} + \frac{\phi_2}{2c_2})$, or at ${A}^* = \frac{\phi_1}{2c_1} + \frac{\phi_2}{2c_2}$ depending on whether $c \geq \max_i c_i$, $c \in (\min_i c_i, \max_i c_i)$, or $c \leq \min_i c_i$, respectively.
\end{proof}
We note in the result above that the optimal assignment is completely deterministic when the commander's cost is low ($c \leq c_{-k}$), and is randomized on both marginals if the cost is sufficiently high ($c\geq c_k$).

\begin{proof}[Proof of Theorem \ref{thm:W_theorem}, second part]
	By comparing the characterizations from Propositions \ref{prop:commander_det} and \ref{prop:commander_costs}, we identify the four-fold improvement in the regime specified in the statement (first case in Prop. \ref{prop:commander_det}, last case in Prop. \ref{prop:commander_costs}). Outside of this regime, the improvement factor is less than four.
\end{proof}

\end{appendices}

\bibliography{sources}

\end{document}